\documentclass[twocolumn]{svjour3}
\usepackage{verbatim}
\usepackage{times}
\usepackage{graphicx, epsfig, color, txfonts}
\usepackage{algorithm}
\usepackage[noend]{algorithmic}
\usepackage{setspace}

\newcommand{\num}{0.050in}
\newcommand{\tnum}{\num*2}
\newtheorem{observation}{Observation}

\renewcommand{\algorithmicrequire}{\textbf{Input:}}
\renewcommand{\algorithmicensure}{\textbf{Output:}}
\renewcommand{\algorithmiccomment}[1]{// #1}
\algsetup{indent=2em}

\newcommand{\FSM}{DFSM}
\newcommand{\e}{e}
\newcommand{\FSMs}{DFSMs}
\newcommand{\efr}{\end{flushright}}
\newcommand{\bfr}{\begin{flushright}}

\newcommand{\h}{\hspace*{0.15in}}
\newcommand{\RCP}{{RCP}}
\newcommand{\State}{\bigtriangleup s}
\newcommand{\Event}{\bigtriangleup e}
\newcommand{\fusion}[2]{($#1$, $#2$)-fusion}
\newcommand{\decomp}[2]{(#1,#2)-event decomposition}

\newcommand{\s}[3]{#1^{#2}_{#3}}
\newcommand{\w}{d_{min}} %% min distance

\newcommand{\recover}{recovery agent}

\newcommand{\myfloor}[1]{\lfloor #1 \rfloor}
\newcommand{\m}{\mathcal}
\newcommand{\col}{\textcolor{black}}

\begin{document}
\title{Fault Tolerance in Distributed Systems using Fused State Machines}
\author{Bharath Balasubramanian \and Vijay K. Garg \thanks{*This research was supported in part by the NSF Grants
CNS-0718990, CNS-0509024, CNS-1115808 and Cullen Trust for Higher
Education Endowed Professorship.
}
}
\institute{Bharath Balasubramanian 
\at EDGE Lab, Dept. of Electrical Engineering,\\
 Princeton University, \\
Engineering Quadrangle, Olden Street, \\
Princeton, NJ 08544. \\
Tel.: +1 512 239 8104, \email{bbharath@utexas.edu}. \and
Vijay K. Garg 
              \at Parallel and Distributed Systems Laboratory, \\
Dept. of Electrical and Computer Engineering, \\
The University of Texas at Austin, \\
1 University Station, C0803, \\
Austin, TX 78712-0240.\\
 Tel.: +1 512 471 9424, \email{garg@ece.utexas.edu}. 
}
\maketitle

\begin{abstract}
Replication is a standard technique for fault tolerance in distributed systems modeled as
deterministic finite state machines (\FSMs{} or machines). To correct $f$ crash or $\lfloor f/2
\rfloor$ Byzantine faults among $n$
different machines, replication requires $nf$ additional backup machines. We 
present a solution called \emph{fusion} that requires just $f$ additional backup machines.
First, we build a framework for fault tolerance in \FSMs{} based on the
notion of Hamming distances. We introduce the concept of an \fusion{f}{m},
which is a set of $m$ backup machines that can correct $f$ crash faults or $\lfloor f/2 \rfloor$
Byzantine faults among a given set of machines.
Second, we present an algorithm to generate an \fusion{f}{f} for a given set of machines. We ensure
that our backups are efficient in terms of the size of their state and event sets. Third, we use \emph{locality sensitive hashing} 
for the detection and correction of faults that incurs almost the same overhead as that for
replication. We detect Byzantine faults with time complexity $O(n f)$ on average while we correct
crash and Byzantine faults with time complexity $O(n  \rho  f)$ with high
probability, where $\rho$ is the average state reduction
achieved by fusion. Finally, our evaluation of  
fusion on the widely used MCNC'91 benchmarks for \FSMs{} show that the average state space savings
in fusion (over replication) is 38\% (range 0-99\%). To demonstrate the practical use
of fusion, we describe its potential application to the MapReduce framework. Using a simple case
study, we compare replication and fusion as applied to this framework. While a pure
replication-based solution requires
1.8 million map tasks, our fusion-based solution requires only 1.4 million map tasks with minimal overhead
during normal operation or recovery. Hence, fusion results in
considerable savings in state space and other resources such as the power needed to run the backup
tasks.  
%while the average event-reduction is 4\% (range 0-45\%). 
%We show that for small values of $n$ (for most practical systems, $n<10$) and
%$\rho$ (average value of $\rho < 2$ in our experiments), this results in almost no overhead as
%compared to replication. 
%In this paper, we first propose a fundamental problem regarding \FSMs{}, independent of fault
%tolerance, that has not been explored in the literature so far: 
%Given a machine $M$, with a set of states and a set of events, can we \emph{replace} it with machines
%each containing fewer events than $M$? To formalize this we define a \decomp{$k$}{$e$} of a given
%machine $M$, that is a set of $k$ machines each
%with at least $e$ events fewer than the event set of $M$, that acting in parallel, are
%equivalent to $M$. 
%We present an algorithm to generate such machines with time
%complexity $O(|X_M|^3|\Sigma_M|^e)$, where $X_M$ is the set of states and $\Sigma_M$ the set of
%events of $M$. Second, we use our event decomposition
%algorithm to generate fused backups that can correct faults among a given set of machines. We show
%that these backups are \emph{minimal} w.r.t the number of states they
%contain and the number of events in their event set. 
%Further, the average savings in time by the incremental approach for generating the fusions (over
%the non-incremental approach) is 8\%. 
%We also suggest an incremental algorithm for the generation of fused backups that improves the time complexity by a
%factor of $\rho^n$, 
\end{abstract}
\keywords{Distributed Systems, Fault Tolerance, Finite State Machines, Coding Theory, Hamming Distances.} 
\section{Introduction}\label{secIntro}
Distributed applications often use deterministic finite state machines (referred to as \FSMs{} or machines) 
to model computations such as regular expressions for pattern detection,
syntactical analysis of documents or mining algorithms for large data sets. These machines executing
on distinct distributed processes are often prone to faults. Traditional solutions to this problem
involve some form of replication. To correct $f$ crash
faults \cite{Sch84} among $n$ given machines (referred to as \emph{primaries}), $f$ copies of each primary are maintained
\cite{Lamp78Reliable,fathi04Replication,schneider90implementing}. If the backups start from the same initial state
as the corresponding primaries and act on the same events, then in the case of faults, the state of
the failed machines can be recovered from one of the remaining copies. These backups can also
correct $\myfloor{f/2}$ Byzantine faults \cite{LaSh82}, where the processes lie about the state of the machine,
since a majority of truthful machines is always available. This approach, requiring $nf$ total
backups, is expensive
both in terms of the state space of the backups and other resources such as the power needed to run
these backups.  
\begin{figure*}[ht]
\centering
\includegraphics[scale=0.50]{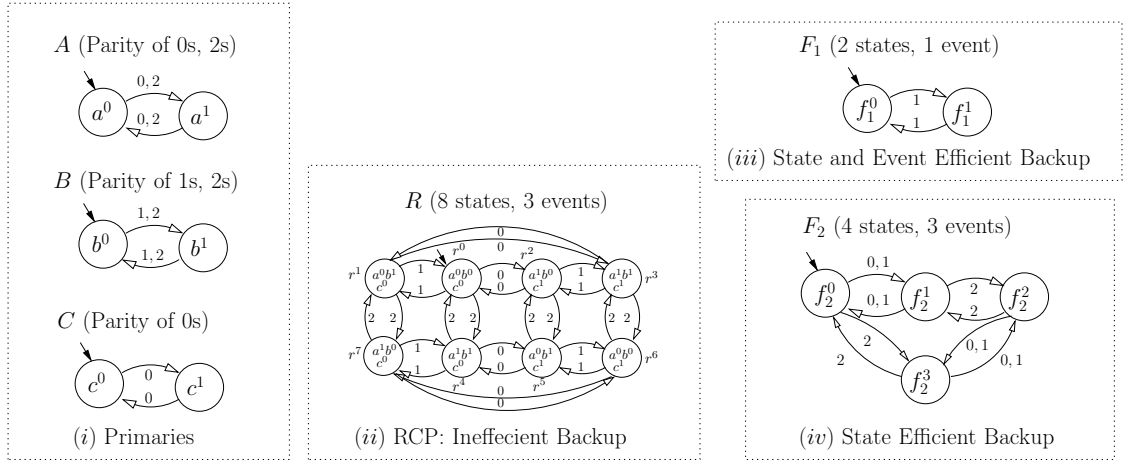}
\caption{Correcting one crash fault among $\{A,B,C\}$ using just one additional backup rather than
three backups required by replication.}
\label{figMainExample}
\end{figure*}

Consider a distributed application that is searching for three different 
string patterns in a file. These string patterns or regular expressions are usually modeled as
\FSMs{}.  Consider the state machines $A$, $B$ and $C$ shown in
Fig. \ref{figMainExample}. A state machine in our system consists of a finite set of states
and a finite set of events. On application of an event, the state machine transitions to the next state based on
the state-transition function. For example, machine $A$ in Fig. \ref{figMainExample} contains the
states $\{a^0,a^1\}$, events $\{0,2\}$ and the initial state, shown by the dark ended arrow, is $a^0$.
The state transitions are shown by the arrows from one state to another. Hence, if $A$ is in state
$a^0$ and event $0$ is applied to it, then it transitions to state $a^1$. 
In this example, $A$ checks the parity of $\{0,2\}$ and so, if it is in state $a^0$, then an even
number of $0s$ or $2s$ have been applied to the machine and if it is in state $a^1$, then an odd
number of the inputs have been applied.  Machines $B$ and $C$ check for the parity of $\{1,2\}$ and
$\{0\}$ respectively. 

To correct one crash fault among these machines, replication requires a copy
of each of them, resulting in three backup machines, consuming total state space of
eight ($2^3$). Another way of looking at replication in \FSMs{} is by constructing a backup machine that is the
\emph{reachable cross product} or $\RCP$ (formally defined in section \ref{secRCP}) of the original
machines. As shown in Fig. \ref{figMainExample}, each state of the $\RCP$, denoted by $R$, is a
tuple, in which the elements corresponds to the states of $A$, $B$ and $C$ respectively. Let each of the machines $A$, $B$, $C$
and $R$ start from their initial state. If some event sequence (generated by the
client/environment) $0 \rightarrow 2
\rightarrow 1$ is applied on 
these machines, then the state of $R$, $A$, $B$ and $C$ are $r^6=\{a^0b^0c^1\}$, $a^0$, $b^0$ and
$c^1$ respectively. Here, even if one of the primaries crash, using the state of $R$, we can
determine the state of the crashed primary. Hence, the
$\RCP$ is a valid backup machine.  

However, using the $\RCP$ of the primaries as a backup has two major disadvantages: $(i)$ Given $n$ primaries each
containing $O(s)$ states, the number of states in the $\RCP$ is $O(s^n)$, which is \emph{exponential} in the number of
primaries. In Fig.
\ref{figMainExample},  $R$ has eight states. $(ii)$ The event set of the $\RCP$ is the union of the event sets of the primaries. In
Fig. \ref{figMainExample}  while $A$, $B$ and $C$ have only two, two and one event respectively in their
event sets, $R$ has
three events. This translates to increased load on the backup. Can we generate backup machines that
are more efficient than the $\RCP$ in terms of states and events? 

%Since $F_1$ is in state $f_1^1$, there are an odd number of $1$s in the
%input sequence. Since $B$ is in state $b^0$, there are an even number of $1$s or $2$s in the input
%sequence. Hence, there are an odd number of $2$s in the input sequence. The state of $A$ tells us
%that there have been odd number of $0$s or $2$s 	
%
Consider $F_1$ shown in Fig. 
\ref{figMainExample}. If the event sequence $0 \rightarrow 0 \rightarrow 1 \rightarrow 2$ is applied the machines, 
$A$, $B$, $C$ and $F_1$, then they will be in states $a^1$, $b^0$, $c^0$ and $f^1_1$. 
Assume a crash fault in $C$. Given the parity of $1$s (state of
$F_1$) and the parity of 1s or 2s
(state of $B$), we can first determine the parity of 2s. Using this, and the parity of
0s or 2s (state of $A$), we can determine the parity of 0s (state of $C$). Hence, we can
determine the state of $C$ as $c^0$ using the states of $A$, $B$ and $F_1$. This argument can be extended to correcting
one fault among any of the machines in $\{A,B,C,F_1\}$. This
approach consumes fewer backups than replication (one vs. three), fewer states 
than the $\RCP$ (two states vs. eight states) and fewer number of events than the $\RCP$ (one event vs. three
events). How can we generate
such a backup for any arbitrary set of machines? In Fig. \ref{figMainExample}, can $F_1$ and $F_2$
correct two crash faults among the primaries? Further, how do we correct the faults? In this
paper, we address such questions through the following contributions: 

\paragraph{Framework for Fault Tolerance in \FSMs{}} 
We explore the idea of a fault graph and use that to define the minimum Hamming
distance \cite{hamming50} for a set of machines. Using this framework, we can specify the
exact number of crash or Byzantine faults a set of machines can correct. Further, we introduce 
the concept of an \emph{\fusion{f}{m}} which is a set of $m$ machines that can correct $f$ crash
faults, detect $f$ Byzantine faults or correct
$\lfloor f/2 \rfloor$ Byzantine faults. We refer to the machines as \emph{fusions} or \emph{fused
backups}. In Fig. \ref{figMainExample}, $F_1$ and $F_2$ can correct
two crash faults among $\{A,B,C\}$ and hence $\{F_1,F_2\}$ is a \fusion{2}{2} of $\{A,B,C\}$. Replication is just a special case of
\fusion{f}{m} where $m = nf$. We prove properties on the \fusion{f}{m} for a given set of
primary machines including lower bounds for the existence of such fusions.

\paragraph{Algorithm to Generate Fused Backup Machines} 
Given a set of $n$ primaries we present an algorithm that generates an \fusion{f}{f} corresponding
to them, i.e., we generate a set of $f$ backup machines that can correct $f$ crash or $\lfloor f/2
\rfloor$ Byzantine faults among them. We show that our backups are efficient in terms of: $(i)$ The number
of states in each backup $(ii)$ The number of events in each backup $(iii)$ The \emph{minimality}
(defined in section \ref{secFusion}) of the entire
set of backups in terms of states. Further, we show that if our algorithm does not achieve state and
event reduction, then no solution with the same number of backups 
achieves it. Our algorithm has time complexity polynomial in $N$, where $N$ is
the number of states in the $\RCP$ of the primaries. We present an incremental approach to this
algorithm that improves the time complexity
by a factor of $O(\rho^n)$, where $\rho$ is the average state savings achieved by fusion. 

%In \cite{OgaBhar09}, the
%algorithm for the correction of  crash and Byzantine faults, has time complexity  $O(n^2  \rho + n
%\rho  f+s^n)$, where $n$ is the number of primaries, $f$ is the number of crash faults, $s$ is the
%maximum number of states among primaries and $\rho$ is the average state savings achieved by fusion.
%

\paragraph{Detection and Correction of Faults} We present a
Byzantine detection algorithm with time complexity $O(nf)$ on average, which is the same as the time
complexity of detection for replication. Hence, for a system that needs to periodically detect
liars, fusion causes no additional overhead.  We reduce the problem of fault correction to one of
finding points within a certain Hamming distance of a given query point in $n$-dimensional space and
present algorithms to correct crash and Byzantine faults with time complexity $O(n  \rho  f)$ with
high probability (w.h.p). The time complexity for crash and Byzantine correction in replication is
$O(f)$ and $O(nf)$ respectively. Hence, for small values of $n$ and $\rho$, fusion causes
almost no overhead for recovery.  Table \ref{tabNotation} describes the main symbols used in this
paper, while Table \ref{tableComparison} summarizes the main results in the paper through a comparison
with replication.% (detailed explanation in the sections of the paper). 

%(which is strictly
%better than the version in \cite{OgaBhar09}). 
\begin{table*}[ht]\centering
\caption{Symbols/Notation used in the paper}
{\small
\begin{tabular}{|c|c||c|c|}
\hline
$\mathcal{P}$& Set of primaries & $n$ & Number of primaries\\
\hline
$\RCP$ & Reachable Cross Product & $N$ & Number of states in the \RCP\\
\hline
$f$ & No. of crash faults & $s$ & Maximum number of states among primaries\\
\hline
$\mathcal{F}$ & Set of fusions/backups& $\rho$ & Average State Reduction in fusion\\
\hline 
$\Sigma$ & Union of primary event-sets& $\beta$ & Average Event Reduction in fusion\\
\hline
\end{tabular}\label{tabNotation}
}
\end{table*}
%\removespace{0.2}

\begin{table*}[htp]
\centering
\caption{Replication vs. Fusion (Columns 2 and 3 for $f$ crash faults, 4 and 5 for $f$ Byzantine faults)}
{\small
\begin{tabular}{|c|p{1.0in}|p{1.0in}||p{1.0in}|p{1.0in}|}
\hline
{}& \bf{Rep-Crash}& \bf{Fusion-Crash} & \bf{Rep-Byz}& \bf{Fusion-Byz} \\
\hline
{Number of Backups}& $nf$ & $f$ & $2nf$ &  $2f$ \\
\hline
{Backup State Space}& $s^{nf}$ & $(s^n/\rho)^{f}$ & $s^{2nf}$ & $(s^n/\rho)^{2f}$ \\
\hline
{Average Events/Backup}& $|\Sigma|/n$ & $ |\Sigma|/\beta$ & $|\Sigma|/n$ & $ |\Sigma|/\beta$  \\
\hline
%{Normal Operation Time} & $O(1)$ & $O(1)$ & $O(1)$ & $O(1)$\\
%\hline
{Fault Detection Time}& $O(1)$ &  $O(1)$ & $O(nf)$  & $O(nf)$ (on avg.)\\
\hline
{Fault Correction Time}& $O(f)$ & $O(n \rho f)$ w.h.p & $O(nf)$ & $O(n \rho f)$ w.h.p \\
\hline
{Fault Detection Messages}& $O(1)$ & $O(1)$ & $2nf$ & $n+f$ \\
\hline
{Fault Correction Messages}& $f$ & $n$ & $n+2f$ & $n+f$ \\
\hline
{Backup Generation Time Complexity}& $O(nsf)$  & $O(s^n|\Sigma|  f/\rho^n)$ & $O(nsf)$  & $O(s^n|\Sigma|  f/\rho^n)$  \\
\hline
\end{tabular}\label{tableComparison}\\
}
\end{table*}
%
%{\small
%\begin{table*}[htp]
%\caption{Fusion vs. Replication ($n$ primaries, $O(s)$ states each, $f$ faults, $|\Sigma|$ total
%events, average state reduction $\rho$)}
%\centering
%{\small
%\begin{tabular}{|c|c|c|}
%\hline
%{}& \bf{Replication}& \bf{Fusion} \\
%\hline
%{Backups for Crash Faults} & $nf$ & $f$  \\
%\hline
%{Backups for Byzantine Faults} & $2nf$ & $2f$  \\
%\hline
%{Backup Space}& $O(s^{nf})$ &  $O((s/\rho)^{nf})$  \\
%\hline
%{Backup Generation Time Complexity}& $O(nsf)$  & $O(s^n 
%|\Sigma|  f/\rho^n)$  \\
%\hline
%{Normal Operation Time}& $O(1)$ &  $O(1)$  \\
%\hline
%{Maximum Events/Backup}& Maximum Events/primary& Minimal for $f$
%backups\\
%\hline
%{Byzantine Detection Time Complexity}& $O(nf)$& $O(nf)$ on average\\
%\hline
%{Crash Correction Time Complexity}& $\theta(f)$& $O(n \rho f)$ w.h.p\\
%\hline
%{Byzantine Correction Time Complexity}& $O(nf)$& $O(n \rho f)$ w.h.p\\
%\hline
%\end{tabular}\label{tableComparison}\\
%}
%%\vspace{-3mm}
%\end{table*}
%}
%In \cite{OgaBhar09}, we evaluated fusion on simple
%examples such as counters and dividers. 
\paragraph{Fusion-based Grep in the MapReduce Framework} 
To illustrate the practical use of fusion, we consider its potential application to the
\emph{grep} functionality of the MapReduce framework \cite{Dean2008}. The MapReduce framework is a
prevalent solution to model large scale distributed computations. The grep functionality is used in
many applications that need to identify patterns in huge textual data such as data mining, machine
learning and query log analysis. Using a simple
case study, we show that a pure replication-based approach for fault tolerance
needs 1.8 million map tasks while our fusion-based solution requires
only 1.4 million map tasks. Further, we show that our approach causes minimal overhead during normal
operation or recovery.   

\paragraph{Fusion-based Design Tool and Experimental Evaluation} We provide a Java design tool based on our
fusion algorithm, that takes a set
of input machines and generates fused backup machines corresponding to them. We evaluate our fusion algorithm on 
the MCNC'91 \cite{Yang91logicsynthesis} benchmarks for \FSMs{}, that are widely used in the fields
of logic synthesis and circuit design. 
Our results show that the average state space savings
in fusion (over replication) is 38\% (range 0-99\%), while the average event-reduction is 4\% (range 0-45\%). 
Further, the average savings in time by the incremental approach for generating the fusions (over
the non-incremental approach) is 8\%. 

\col{In section \ref{secFsmModel}, we specify the system model and assumptions of our work. In section
\ref{secFsmFaultTolerance} we describe the theory of our backup or fusion
machines. Following this, we present algorithms to generate these fusion machines in section
\ref{secSpaceEventFusions}. In section \ref{secFaults} we present the
algorithms for the detection and correction of faults in a system with primary and fusion machines.
Sections \ref{secMapReduce} and \ref{secEvaluate} deal with the practical aspects and experimental
evaluation of fusion. In section \ref{secMacOutLat}, we consider potential solutions to this
problem, outside the framework of this paper. Section \ref{secRelatedWork} covers the related work
in this area. Finally, we summarize our work and discuss future extensions in section \ref{secConc}.}

\section{Model}\label{secFsmModel}
The \FSMs{} in our system execute on separate distributed processes. We assume
loss-less FIFO communication links  with a strict upper bound
on the time taken for message delivery.  Clients of the state machines issue the events (or commands) to the
concerned primaries and backups. \col{For simplicity, we assume that there is a single client issuing the
events to the machines. This along with FIFO links ensures that all machines act on the events in
the same relative order. This can be extended to multiple clients using standard total
order broadcast mechanisms present in the literature \cite{Defago2004,MelliarSmith1990}. 
%that achieves these two conditions, such as Lamport's logical clocks \cite{Lamp:HappenBefore} or synchronized real-time clocks
%\cite{crist89,Ramanathan1990}. 
}

The \emph{execution state} of a machine is the current state in
which it is executing. Faults in our system are of two types: crash faults,
resulting in a loss of the execution state  of the machines and Byzantine faults resulting in an
arbitrary execution state. We assume that the given set of primary machines cannot correct a single
crash fault amongst themselves. When faults are detected by a trusted recovery agent using timeouts (crash
faults) or a detection algorithm (Byzantine faults) no further events are sent by any client to
these machines. Assuming the machines have acted on the same sequence of events, 
the \recover{} obtains their states, and recovers the correct
execution states of all faulty machines. 

%Note that, while replication is a
%fault-masking technique, in fusion, we need to retrieve the state of all the available
%machines and use our algorithms to recover the state of the failed machines.  
%In replication, the replicas for each machine are identical to the given data
%structure. In fusion, the backup copies are not identical to the given state machines 
%and hence, we make a distinction between the given state machines, referred to as \emph{primaries} and
%the backup state machines, referred to as \emph{backups}. Henceforth in this paper, we assume that
%we are given a set of primary state machines among which we need to correct faults. 
%Replication requires $f$ additional copies of each primary ($f+1$ replicas),
%while fusion only requires $f$ additional backups in total. 
%

%Since we assume a trusted recovery agent, the work on
%consensus in the presence of Byzantine faults \cite{FLP85,PeaseLamp80}, does not apply to our
%paper. 
%In the following section, we summarize the relevant
%concepts and results introduced in our previous work. 
%

\section{Framework for Fault Tolerance in \FSMs{}}\label{secFsmFaultTolerance}

In this section, we describe the framework using which we can
specify the exact number of crash or Byzantine faults that any set of machines can correct. Further,
we introduce the concept of an \fusion{f}{m} for a set of primaries that is a set of
machines that can correct $f$ crash faults, detect $f$ Byzantine faults and correct $\lfloor f/2
\rfloor$ Byzantine faults. 

\subsection{\FSMs{} and their Reachable Cross Product}\label{secRCP}

A \FSM{}, denoted by ${A}$, consists of a set of states $X_A$, set of events $\Sigma_A$, transition
function $\alpha_A:X_A \times \Sigma_A \rightarrow X_A$ and initial state $\s{a}{0}{}$.  The size of
$A$, denoted by $|A|$ is the number of states in $X_A$. A state, $s \in X_A$, is \emph{reachable} iff
there exists a sequence of events, which, when applied on the initial state $\s{a}{0}{}$, takes the
machine to state $s$. %This is denoted by $s= \alpha^k(\s{a}{0}{})$, where $\alpha^k$ denotes a
%sequence of $k$ operations, $\alpha^1, \ldots, \alpha^k$ applied to the initial state $\s{a}{0}{}$.
%Our model assumes that all the states corresponding to the machines are reachable. 
Consider any two machines, $A$ $(X_A,\:\Sigma_A,\:\alpha_A,\:\s{a}{0}{})$ and $B$
$(X_B,\:\Sigma_B,\:\alpha_B,\:\s{b}{0}{})$. Now construct another machine which consists of all the
states in the product set of $X_A$ and $X_B$  with the transition function $\alpha '
(\{a,b\},\sigma) = \{ \alpha_A(a,\sigma) , \alpha_B(b,\sigma)\}$ for all $\{a,b\} \in X_A \times
X_B$ and $\sigma \in \Sigma_A \cup \Sigma_B$. This machine $( X_A \times X_B,\:\Sigma_A \cup
\Sigma_B,\:\alpha ',\: \{ \s{a}{0}{}, \s{b}{0}{} \} )$  may have states that are not reachable from
the initial state $\{ \s{a}{0}{}, \s{b}{0}{} \}$. If all such unreachable states are pruned, we get
the \emph{reachable cross product} of $A$ and $B$. In 
Fig. \ref{figMainExample}, $R$ is the reachable cross product of 
$A$, $B$ and $C$. Throughout the paper, when we just say $\RCP$, we refer to the reachable cross
product of the set of primary machines. Given a set of primaries, the number of states in its 
$\RCP$ is denoted by $N$ and its event set, which is the union of the event sets of the primaries is
denoted by $\Sigma$. 

As seen in section \ref{secIntro}, given
the state of the $\RCP$, we can determine the state of each of the
primary machines and vice versa. However, the $\RCP$ has states exponential in $n$ and an event set  that is the union of all primary
event sets. \emph{Can we generate machines that contains fewer states
and events than the $\RCP$?} In the
following section, we first define the notion of order and the `less than or equal to' $(\leq)$ relation among machines.

%For example, in Fig. \ref{figCrossProduct}, let each of the machines $A$, $B$, $C$
%and $\RCP(\m{P})$ start from their initial state. If some event sequence $0 \rightarrow 2
%\rightarrow 1$ is applied on 
%these machines, then the state of $\RCP(\m{P})$ is $r^6=\{a^0b^1c^1\}$. Clearly, the state of the
%$A$, $B$ and $C$ after they act on the same event sequence is $a^0$, $b^1$ and $c^1$. Hence, the
%$\RCP$ of the primaries is always a valid backup machine.  
%

%\begin{observation}
%Given a set of $n$ primaries $\m{P}$, $f$ copies of the $\RCP(\m{P})$ can correct $f$ crash faults 
%faults among $\m{P}$.
%\end{observation}
%

\subsection{Order Among Machines and their Closed Partition Lattice}\label{secOrder}

\col{Consider a \FSM{}, $A =
(X_A,\Sigma,\alpha_{A},\s{x}{0}{A})$. A \emph{partition} $P$, on the state set $X_A$ of $A$ is the set $\{B_1,\ldots,B_k\}$,  of disjoint subsets of the
state set $X_A$, such that $\bigcup^k_{i=1} B_i = X_A$ and $B_i \cap B_j = \phi$ for $i \neq j$
\cite{LeeYann2002}. An element $B_i$ of a partition is called a \emph{block}. 
A partition, $P$, is said to be closed if each event, $\sigma \in \Sigma$, maps a block of $P$ into
another block. A closed partition $P$, corresponds to a distinct machine. 
Given any machine $A$, we can partition its state space such that the transition function $\alpha_A$, maps each block of the
partition to another block for all events in $\Sigma_A$ \cite{HartSteBook,LeeYann2002}.}

In other
words, we combine the states of $A$ to generate machines that are consistent with the transition
function. We refer to the set of all such closed partitions as the closed
partition set of $A$. In this paper, we discuss the closed partitions
corresponding to the $\RCP$ of the primaries. In Fig. \ref{figStEvRedLattice}, we show the closed
partition set of the $\RCP$ of $\{A,B,C\}$ (labeled $R$). Consider machine $M_2$ in Fig.
\ref{figStEvRedLattice}, generated by combining the states $r^0$ and $r^2$ of $R$. Note that, on
event 1, $r^0$ transitions to
$r^1$ and $r^2$ transitions to $r^3$. Hence, we need to combine the states $r^1$ and $r^3$. Continuing
this procedure, we obtain the combined states in $M_2$. Hence, we have \emph{reduced} the $\RCP$ to
generate $M$. By combining different pairs of states and
by further reducing the machines thus formed, we can construct the entire closed partition set of $R$. 
%On event 0, $\{r^0,r^2\}$
%self-transitions to $\{r^0,r^2\}$ (self transitions not shown). 

\begin{figure*}[htb] 
\centerline{ 
\scalebox{0.40}{ 
\includegraphics{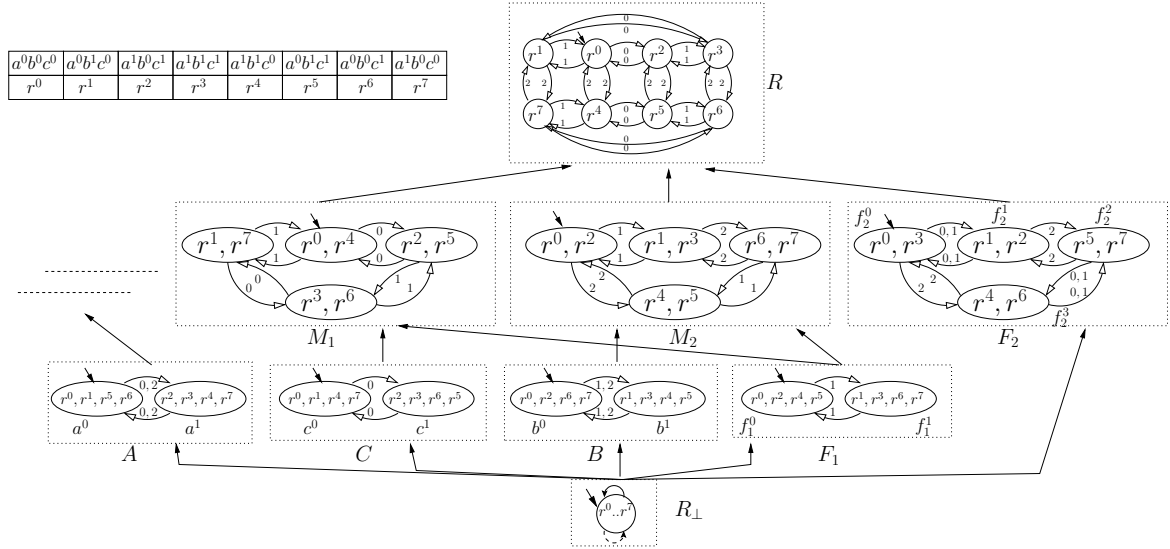} } } 
\caption{Set of machines less than $R$ (all machines not shown due to space constraints).}
\label{figStEvRedLattice} \end{figure*}

We can define an
order ($\leq$) among any two machines $P$ and $Q$ in this set as follows: $P \leq Q$, if each block
of $Q$ is contained in a block of $P$ (shown by an arrow from $P$ to $Q$). Intuitively, given the state of
$Q$ we can determine the state of $P$. Machines  $P$ and
$Q$ are incomparable, i.e., $P || Q$, if $P \not< Q$ and $Q \not< P$. In Fig. \ref{figStEvRedLattice}, $F_1 < M_2$, while $M_1 || M_2$. It can be seen that the set of all closed partitions
corresponding to a machine, form a lattice under the $\leq$ relation \cite{HartSteBook}. We saw in
section \ref{secRCP} that given the state of the primaries, we can determine the state of the $\RCP$
and vice versa. Hence, the primary machines are always part of the closed partition set of the
$\RCP$ (see $A$, $B$ and $C$ in Fig. \ref{figStEvRedLattice}). 

%Note that, both $F_1$ and $F_2$ in Fig. \ref{figMainExample}, can be found in Fig.
%\ref{figStEvRedLattice}. 
Among the machines shown in
Fig. \ref{figStEvRedLattice}, some of them, like $F_2$ (4 states, 3 events) have reduced states, while some like $M_1$ (4 states, 2
events) and $F_1$ (2 states, 1 event) have
both reduced states and events as compared to $R$ (8 states, 3 events). \emph{Which among these
machines can act as backups?} In
the following section, we describe the concept of fault graphs and their Hamming distances 
to answer this question.

\subsection{Fault Graphs and Hamming Distances}\label{secGraphHamming}

We begin with the idea of a \emph{fault graph} of a set of machines  $\mathcal{M}$, for a machine
$T$, where all machines in $\mathcal{M}$ are less than or equal to $T$. This is a weighted graph and
is denoted by $G(T,\mathcal{M})$. The fault graph is an indicator of the capability of the set of
machines in $\mathcal{M}$ to correctly identify the current state of $T$. As described in the
previous section, since all the machines in $\mathcal{M}$ are less than or equal to $T$, the set of
states of any machine in $\mathcal{M}$ corresponds to a closed partition of the set of states of
$T$. Hence, given the state of $T$, we can determine the state of all the machines in $\m{M}$ and
vice versa. 
%Every state of $T$ corresponds to a node of the fault graph $G(T,\mathcal{M})$ and the graph is
%totally connected. The weight of the edge between nodes corresponding to states $r^i$ and $r^j$ of
%the fault graph is the number of machines in $\mathcal{M}$ that have states $r^i$ and $r^j$ in
%distinct blocks. 

\begin{definition}\label{defFaultGraph} (Fault Graph)
Given a set of machines $\mathcal{M}$ and a machine $T= (X_T,\Sigma_T,\alpha_T,\s{t}{0}{})$ such that
$\forall M \in \mathcal{M}: M \leq T$, the fault graph $G(T, \mathcal{M})$ is a \emph{fully
connected weighted graph} where,
\begin{itemize}
\item Every node of the graph corresponds to a state in $X_T$
\item  The weight of the edge $(t^i,t^j)$ between two nodes, where $t^i,t^j \in X_T$, is the number
of machines in $\mathcal{M}$ that have states $t^i$ and $t^j$ in distinct blocks 
\end{itemize}
\end{definition}

 We construct the fault graph $G(R, \{A\})$, referring to Fig. \ref{figStEvRedLattice}. $A$ has
two states, $a^0=\{r^0, r^1,r^5,r^6\}$ and  $a^1=\{r^2,r^3,r^4,r^7\}$.  Given just the current state
of $A$, it is possible to determine if $R$ is in state $r^0$ or $r^2$ (exact) or one of $r^0$ and
$r^1$ (ambiguity). Here, $A$ distinguishes between the $(r^0,r^2)$ but not between $(r^0,r^1)$. 
Hence, in the fault graph $G(R, \{A\})$ in Fig. \ref{figFaultGraph} $(i)$,
the edge $(r^0,r^2)$ has weight one, while $(r^0,r^1)$ has weight zero. 
A machine $M \in \m{M}$, is said to \emph{cover} an edge $(t^i,t^j)$ if $t^i$ and $t^j$ lie in
separate blocks of $M$, i.e., $M$ \emph{separates} the states $t^i$ and $t^j$. In Fig. \ref{figStEvRedLattice}, $A$ covers 
$(r^0,r^2)$. In Fig.
\ref{figOldClosedPartitionLattice} and \ref{figOldFaultGraph} of the Appendix, we show an example of the closed
partition set and fault graphs for a different set of primaries. 

Given the states of  $|\mathcal{M}| - x$  machines in $|\mathcal{M}|$, it is always possible to
determine if $T$ is in state $t^i$ or $t^j$ iff the weight of the edge  $(t^i,t^j)$ is greater than
$x$. Consider the graph shown in Fig. \ref{figFaultGraph} $(ii)$. Given the state of any two 
machines in $\{A,B,C\}$, we can determine if $R$ is in state $r^0$ or $r^2$, since the weight of that
edge is greater than one, but cannot do the same for the edge $(r^0, r^1)$, since the weight of the
edge is one. In coding theory \cite{BerleCoding68,PetersonCodes72}, the concept of Hamming distance \cite{hamming50} is widely used to specify the
fault tolerance of an erasure code. If an erasure code has minimum Hamming distance greater than $d$, then it
can correct $d$ erasures or $\lfloor d/2 \rfloor$ errors. To understand the fault
tolerance of a set of machines, we define a similar notion of distances for the fault graph.

\begin{definition}(distance) Given a set of machines $\mathcal{M}$
and their reachable cross product $T$ $(X_T,\Sigma_T,\alpha_T,\s{t}{0}{})$, the  distance between any
two states $t_i, t_j \in X_T$, denoted by $d(t_i,t_j)$, is the weight of the edge $(t_i,t_j)$ 
in the fault graph $G(T, \mathcal{M})$.  The least distance in $G(T, \mathcal{M})$ is denoted by $\w(T, \mathcal{M})$. 
\end{definition}

Given a fault graph, $G(T,\mathcal{M})$, the smallest distance between the nodes in the fault graph
specifies the fault tolerance of $\mathcal{M}$. Consider the graph,
$G(R,\{A , B,C, F_1, F_2\})$, shown in  Fig. \ref{figFaultGraph} $(v)$. Since the smallest distance in
the graph is three, we can remove any two machines from  $\{A , B,C, F_1, F_2\}$ and still regenerate
the current state of $R$. As seen before, given the state of $R$, we can determine the state of any
machine less than $R$.  Therefore, the set of machines $\{A , B, C, F_1, F_2\}$ can correct two crash
faults. 
\begin{figure*}[!htb]
\centerline{
\scalebox{0.4}{
 \includegraphics{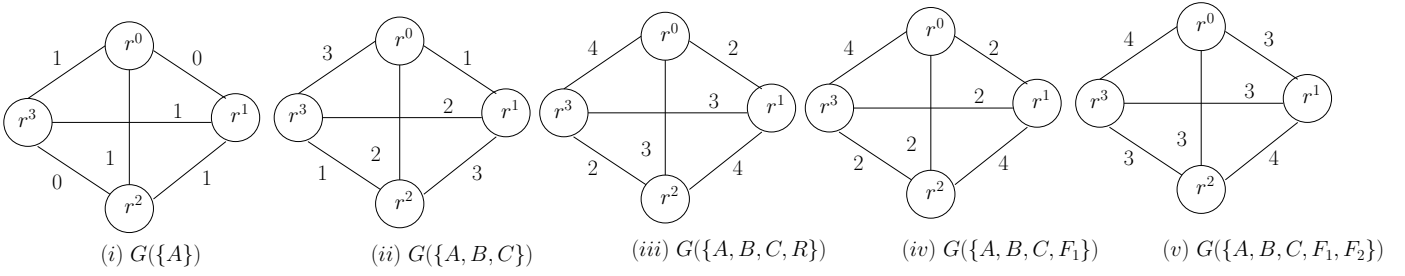}
} 
}
\caption{Fault Graphs, $G (R,\mathcal{M})$, for sets of machines shown in Fig.
\ref{figStEvRedLattice}. For notational convenience, we just label the graphs with $G(\mathcal{M})$. All eight nodes $r^0$-$r^7$ with their edges have not been shown due to
space constraints.}
\label{figFaultGraph}
\end{figure*}

%Based on the fault graph, we define the Hamming distance between any two global states belonging to $X_T$.
%\begin{definition}
%Given a set of machines $\mathcal{M}$ and their reachable cross product $T$ $(X_T,\Sigma,\alpha,\s{t}{0}{})$, the hamming distance between any two states $t_i, t_j \in X_T$, denoted $d(t_i,t_j)$ is the weight of the edge between the nodes corresponding to $t_i$ and $t_j$ in the fault graph $G(T, \mathcal{M})$.
%\end{definition}

\begin{theorem}\label{thFaultTolerance} A set of machines $\mathcal{M}$, can correct up to $f$
crash faults iff $\w(T,\mathcal{M}) > f$, where $T$ is the reachable cross-product of all machines
in $\mathcal{M}$.  \end{theorem}

\begin{proof} $(\Rightarrow)$  Given that $\w(T,\mathcal{M}) > f$, we show that any $\mathcal{M} - f$
machines from $\mathcal{M}$ can accurately determine the current state of $T$, thereby recovering
the state of the crashed machines. Since $\w(T,
\mathcal{M}) > f$, by definition, at least $f+1$ machines separate any two states of $X_T$. Hence,
for any pair of states $(t_i, t_j) \in X_T$, even after $f$ crash failures in $\mathcal{M}$, 
at least one machine remains that can distinguish between $t_i$ and $t_j$. This implies
that it is possible to accurately determine the current state of $T$ by using any $\mathcal{M} - f$
machines from $\mathcal{M}$.

$(\Leftarrow)$  Given that $\w(T,\mathcal{M})
\leq f$, we show that the system cannot correct $f$ crash faults. The condition $\w(T,\mathcal{M}) \leq f$ implies that there exists states $t_i$ and $t_j$ in
$G(T,\mathcal{M})$ separated by distance $k$, where $k \leq f$.  Hence there exist exactly $k$
machines  in $\mathcal{M}$ that can distinguish
between states $t_i,t_j \in X_T$. Assume that all these $k$ machines crash (since $k \leq f$) when
$T$ is in either $t_i$ or $t_j$. Using the states of the remaining machines in $\mathcal{M}$, it is
not possible to determine whether $T$ was in state $t_i$ or $t_j$. Therefore, it is not possible to
exactly regenerate the state of any machine in $\mathcal{M}$ using the remaining machines.

\end{proof}

 %For example if $\w (\mathcal{M}$ for a set of machines $\mathcal{M}$ is $5$ then the system can
 %correct $5$ crash faults.
Byzantine faults may include machines which lie about their state. Consider
the machines $\{A,B,C,F_1,F_2\}$ shown in Fig. \ref{figStEvRedLattice}. From Fig.
\ref{figFaultGraph} $(v)$,  Let the execution states of the machines $A$, $B$,
$C$, $F_1$ and $F_2$ be $$a^0=\{r^0,r^1,r^5,r^6\}, b^1=\{r^1,r^3,r^4,r^5\},
c^0=\{r^0,r^1,r^4,r^7\}$$ $$f_1^0=\{r^0,r^2,r^4,r^5\}, f_2^0=\{r^0,r^3\},$$ respectively. Since
$r^0$ appears four times (greater than majority) among these states, even if there is one liar we can determine that $R$ is
in state $r^0$. But if $R$ is in state $r^0$, then $B$ must have been in state $b^0$ which contains
$r^0$. So clearly, $B$ is lying and its correct state is $b^1$. Here, we can determine the correct state 
of the liar, since $d_{min}(R,
\{A,B,C, F_1,F_2\})=3$, and the
majority of machines distinguish between all pairs of states. 

\begin{theorem}\label{thByzFaultTolerance} A set of machines $\mathcal{M}$, can correct up to $f$
Byzantine faults iff $\w(T,\mathcal{M}) > 2f$, where $T$ is the reachable cross-product of all
machines in $\mathcal{M}$.  \end{theorem}

\begin{proof} $(\Rightarrow)$  Given that $\w(T,\mathcal{M}) > 2f$, we show that any $\mathcal{M} - f$
correct machines from $\mathcal{M}$ can accurately determine the current state of $T$ in spite of
$f$ liars. Since $\w(T, \mathcal{M}) > 2f$, at least $2f+1$ 
machines separate any two states of $X_T$. Hence, for any pair of states $t_i, t_j \in X_T$, 
after $f$ Byzantine failures in $\mathcal{M}$, there will always be at least $f+1$ correct
machines that can distinguish between $t_i$ and $t_j$. This implies that it is possible to
accurately determine the current state of $T$ by simply taking a majority vote.

$(\Leftarrow)$  Given that $\w(T,\mathcal{M})
\leq 2f$, we show that the system cannot correct $f$ Byzantine faults. $\w(T,\mathcal{M}) \leq 2f$ implies that there exists states $t_i, t_j \in X_T$ 
separated by distance $k$, where $k \leq 2f$.  If $f$ among these $k$ machines 
lie about their state, we have only  $k-f$ correct machines remaining. Since, $k-f \leq f$, it is
impossible to distinguish the liars from the truthful machines and regenerate the correct state of
$T$. 
\end{proof}
 
In this paper, we are concerned only with the fault graph of machines w.r.t the $\RCP$ of the
primaries $\m{P}$. For notational convenience,
we use $G(\mathcal{M})$ instead of $G(\RCP, \mathcal{M})$ and $\w(\mathcal{M})$ instead of $\w(\RCP,
\mathcal{M})$. From theorems \ref{thFaultTolerance} and \ref{thByzFaultTolerance}, it is clear that
a set of $n$ machines $\mathcal{P}$, can
correct $(\w(\mathcal{P})-1)$ crash faults and $\lfloor (\w(\mathcal{P})-1) /2 \rfloor$ Byzantine faults. 
Henceforth, we only consider backup machines less than or equal to the $\RCP$ of the primaries.
In the following section, we describe the theory of such backup machines.
 
\subsection{Theory of \fusion{f}{m}} \label{secFusion}

To correct faults in a given set of machines, we need to add backup machines so that the fault
tolerance of the system (original set of machines along with the backups) increases to the desired
value. To simplify the discussion, in the remainder of this paper, unless specified otherwise, we
mean crash faults when we simply say faults. Given a set of $n$ machines $\mathcal{P}$, we add $m$ backup machines $\mathcal{F}$, each less than
or equal to the $\RCP$, such that the set of machines in $\mathcal{P \cup F}$ can correct $f$ faults.
We call the set of $m$ machines in $\mathcal{F}$, an \fusion{f}{m} of $\mathcal{P}$. From theorem
\ref{thFaultTolerance}, we know that, $\w(\mathcal{P \cup F}) > f$.

\begin{definition}(Fusion) \label{defFusion}
Given a set of $n$ machines $\mathcal{P}$, we refer to the set of $m$ machines $\mathcal{F}$, as an \emph{\fusion{f}{m}} of $\mathcal{P}$,  if $\w(\mathcal{P \cup F}) > f$.
\end{definition}

Any machine belonging to $\m{F}$ is referred to as a \emph{fused backup}  or just a \emph{fusion}. 
Consider the set of machines, $\mathcal{P} = \{A,B,C\}$, shown in
Fig. \ref{figMainExample}. From Fig. \ref {figFaultGraph} $(ii)$, $\w(\{A,B,C\}) = 1$. Hence the set
of machines $\m{P}$, cannot correct a single fault. To generate a set of machines $\mathcal{F}$, such that, $\mathcal{P \cup
F}$ can correct two faults, consider Fig. \ref{figFaultGraph} $(v)$. Since 
$\w(\{A,B,C,F_1,F_2\}) = 3$, $\{A,B,C, F_1,F_2\}$ can correct two 
faults. Hence, $\{F_1,F_2\}$ is a \fusion{2}{2} of $\{A,B,C\}$. Note that the set
of machines in $\{A,A,B,B,C,C\}$, i.e., replication, is a \fusion{2}{6} of $\{A,B,C\}$.  

%Based on the values of $f$ and $m$, we discuss three cases of \fusion{f}{m}:
%\begin{itemize}
%\item $f = m$:  In this case, the number of fusion machines equals the number of faults. The set of
%machines in $\{F_1,F_2\}$, shown in Fig. \ref{figStEvRedLattice}, form a \fusion{2}{2} of $\{A,B,C\}$.
% 
%\item $f < m$: The traditional approach of replication is the simplest example for this case. To
%correct two faults in any two machines $\{A,B,C\}$, replication will require two additional copies
%of each of them. Hence, $\{A,A,A,B,B,B\}$ is a \fusion{2}{6} of $\{A,B,C\}$. 
%
%\item $f > m$: From observation \ref{obInherent}, if a system is inherently fault tolerant, then no
%additional machines may be needed to correct faults. In the example shown in Fig.
%\ref{figStEvRedLattice}, let us assume that the original set of machines are $\{A, B,C, F_1\}$. From
%Fig. \ref{figFaultGraph} $(iv)$,
%since $\w(\{A,B,C,F_1\}) = 2$, these machines can correct one fault without any additional machine.
%\end{itemize}
%
Any machine in the set $\{A,B,C, F_1,F_2\}$ can at
most contribute a value of one to the weight of any edge in the graph $G(\{A,B,C,F_1,F_2\})$. Hence, even if we
remove one of the machines, say $F_2$, from this set, $\w(\{A,B,C,F_1\})$ is greater than
one. So $\{F_1\}$ is an \fusion{1}{1} of $\{A,B,C\}$. %Similarly, $\{F_2\}$ is also a
%\fusion{1}{1} of $\{A,B,C\}$. This property is generalized in the following theorem.

\begin{theorem} (Subset of a Fusion) \label{thSubsetOfAFusion} Given a set of $n$ machines
$\mathcal{P}$, and an \fusion{f}{m} $\mathcal{F}$, corresponding to it, any subset $\mathcal{F}'
\subseteq \mathcal{F}$ such that $|\mathcal{F}'| = m-t$  is a  \fusion{f-t}{m-t} when $t \leq
min(f,m)$.  \end{theorem}

\begin{proof}
Since, $\mathcal{F}$ is an \fusion{f}{m} of $\mathcal{P}$, $\w(\mathcal{P \cup F}) > f$.  Any machine, $F \in \mathcal{F}$, can at most
contribute a value of one to the weight of any edge of the graph, $G(\mathcal{P \cup F})$.
 Therefore, even if we remove $t$ machines from
the set of machines in $\mathcal{F}$, $\w(\mathcal{P \cup F}) > f-t$.  Hence, for any subset
$\mathcal{F}' \subseteq \mathcal{F}$, of size $m-t$, $\w(\mathcal{P \cup F'}) > f-t$. This implies
that $\mathcal{F'}$ is an \fusion{f-t}{m-t} of $\mathcal{P}$. 
\end{proof}

It is important to note that the converse of this theorem is not true. In Fig.
\ref{figStEvRedLattice}, while $\{M_2\}$ and
$\{F_1\}$ are \fusion{1}{1}s of $\{A,B,C\}$, since $\w(\{A,B,C,M_2,F_1\}) = 2$, $\{M_2, F_1\}$ is not a
\fusion{2}{2} of $\{A,B,C\}$. We now consider the existence of an \fusion{f}{m} for a given set of machines $\mathcal{P}$.
Consider the existence of a \fusion{2}{1} for $\{A,B,C\}$ in Fig. \ref{figStEvRedLattice}. From Fig. \ref{figFaultGraph} $(ii)$,
$\w(\{A, B,C\}) = 1$. Clearly, $R$ covers each pair of edges in the fault graph. Even if we add $R$ to this set, from Fig. \ref{figFaultGraph}
$(iii)$, $\w(\{A, B,C, R\}) < 3$. Hence, there cannot exist a \fusion{2}{1} for $\{A,B,C\}$.
%We formalize this in the following theorem. 

\begin{theorem}\label{thExistenceFmFusion}(Existence of Fusions)  Given a set of $n$
machines $\mathcal{P}$, there exists an \fusion{f}{m} of $\mathcal{P}$ iff $m + \w(\mathcal{P}) >
f$.  \end{theorem} 
\begin{proof} 
$(\Rightarrow)$
Assume that there exists an \fusion{f}{m} $\mathcal{F}$ for the given set of machines $\mathcal{P}$. 
Since, $\mathcal{F}$ is an \fusion{f}{m} fusion of $\mathcal{P}$, $\w(\mathcal{P \cup F}) > f$. The
$m$ machines in $\mathcal{F}$, can at most contribute a value of $m$ to the weight of each edge in
$G(\mathcal{P \cup F})$. Hence, $m + \w(\mathcal{P})$ has to be greater than $f$.

$(\Leftarrow)$
Assume that $m + \w(\mathcal{P}) >  f$. Consider a set of $m$ machines $\mathcal{F}$, containing $m$
copies of the $\RCP$. These copies contribute exactly $m$ to the weight of each edge in $G(\mathcal{P \cup F})$. Since,
$\w(\mathcal{P}) >  f -m$, $\w(\mathcal{P \cup F}) > f$. Hence, $\mathcal{F}$ is an \fusion{f}{m}
of $\mathcal{P}$. 

\end{proof}

Given a set of machines, we now define an order among \fusion{f}{m}s corresponding to them.

\begin{definition}(Order among \fusion{f}{m}s)
Given a set of $n$ machines $\mathcal{P}$, an \fusion{f}{m} $\mathcal{F} =\{F_1,..F_m\}$, is less
than another \fusion{f}{m} $\mathcal{G}$, i.e, $\mathcal{F} < \mathcal{G}$, iff the machines in
$\mathcal{G}$ can be ordered as $\{G_1,G_2,..G_m\}$ such that $\forall 1\leq i \leq m: (F_i \leq
G_i) \wedge (\exists j: F_j < G_j)$.
\end{definition}

An \fusion{f}{m} $\mathcal{F}$ is \emph{minimal}, if there exists no \fusion{f}{m} $\mathcal{F}'$,
such that, $\mathcal{F}' < \mathcal{F}$. It can be seen that, $$\w(\{A,B,C, M_2,F_2\}) =3,$$
and hence, $\mathcal{F}' = \{M_2,F_2\}$ is a \fusion{2}{2} of $\{A,B,C\}$. We have seen that
$\mathcal{F} = \{F_1, F_2\}$, is a \fusion{2}{2} of $\{A,B,C\}$. From Fig. \ref{figStEvRedLattice},
since $F_1 < M_2$, $\mathcal{F} < \mathcal{F}'$. In Fig. \ref{figStEvRedLattice}, since $R_\bot$
cannot be a fusion for $\{A,B,C\}$, there exists no \fusion{2}{2} less than $\{F_1,F_2\}$. 
Hence, $\{F_1,F_2\}$ is a minimal \fusion{2}{2} of
$\{A,B,C\}$. 

\col{We now prove a property of the fusion machines that is crucial for practical applications. Consider
a set of primaries $\m{P}$ and an \fusion{f}{m} $\m{F}$ corresponding to it. The client sends
updates addressed to the primaries to all the backups as well. We show that events or
inputs that belong to distinct set of primaries, can be received in any order at each of the fused
backups. This eliminates the need for synchrony at the backups.}

Consider a fusion $F \in \m{F}$. Since the states of $F$ are essentially partitions of the state set
of the $\RCP$, the state transitions of $F$ are defined by the state transitions of the $\RCP$. For
example, machine $M_1$ in Fig. \ref{figStEvRedLattice} transitions from $\{r^0,r^2\}$ to
$\{r^1,r^3\}$ on event 1, because $r^0$ and $r^2$ transition to $r^1$ and $r^3$ respectively on
event 1. Hence, if we show that the state of the $\RCP$ is independent of the order in which it
receives events addressed to different primaries, then the same applies to the fusions. 

\begin{theorem}(Commutativity)
The state of a fused backup after acting on a sequence of events, is independent of the order in which the
events are received, as long as the events belong to distinct sets of primaries. 
\end{theorem}
\begin{proof}
We first prove the theorem for the $\RCP$, which is also a valid fused backup. Let the set of primaries be $\m{P}=\{P_1\ldots P_n\}$. Consider an event $e_i$ that belongs to
the set of primaries $\m{S}_i \subseteq \m{P}$. If the $\RCP$ is in state $r$, its next state transition on
event $e_i$ depends only on the transition functions of the primaries in $\m{S}_i$. Hence, the state
of the $\RCP$ after acting on two events $e_a$ and $e_b$ is independent of the order in which these
events are received by the $\RCP$, as long as $\m{S}_a \cup \m{S}_b = \phi$. The proof of the
theorem follows directly from this. 
\end{proof}

So far, we have presented the framework to understand fault tolerance among machines. Given a set of
machines, we can determine if they are a valid set of backups by constructing the fault graph of
those machines. In the following section, we present a technique to generate such backups
automatically.

%Using the analogy of Hamming distances, it is easy to
%see that  the idea of fusions and crash fault tolerance can be extended to Byzantine faults. An
%\fusion{f}{m} can correct $f$ crash faults, detect $f$ Byzantine faults and correct $\lfloor f/2
%\rfloor$ Byzantine faults. %

\section {Algorithm to Generate Fused Backup Machines}\label{secSpaceEventFusions} 

\begin{figure*}[htb] \begin{center}
\fbox{\begin{minipage}[b]  {0.48\linewidth}
{\small
\emph{genFusion}\\
\h{\bf Input}: Primaries $\mathcal{P}$, faults $f$, state-reduction parameter $\State$,\\ 
\h event-reduction parameter $\Event$;\\
\h{\bf Output}: \fusion{f}{f} of $\mathcal{P}$; \\
%\h Identify weakest edges in fault graph $G(\mathcal{P})$; \\
\h$\mathcal{F} \leftarrow \{\}$; \\
\h //Outer Loop\\
\h{\bf for} {$(i =1$ to $f)$ } \\
	\h\h Identify weakest edges in fault graph $G(\mathcal{P} \cup \mathcal{F})$; \\
	\h\h $\mathcal{M} \leftarrow \{RCP(\mathcal{P})\}$; \\
%	 \h /* STEP 1, Event Reduction: Starting from \RCP{} reduce one event at a time till no
%machine in the reduced set increases $\w$ by 1 */\\
\h \h//State Reduction Loop\\
	 \h \h{\bf for} {($j=1$ to $\State$)} \\
		\h\h\h 	$\mathcal{S} \leftarrow \{\}$;\\
		\h\h\h {\bf for} {$(M \in \mathcal{M})$} \\
		   \h\h\h\h $\mathcal{S}= \mathcal{S} \cup \emph{reduceState}(M)$;\\
		\h\h\h $\mathcal{M}$ = All machines in $\mathcal{S}$ that increment $\w(\mathcal{P} \cup
\mathcal{F})$;\\ 
	\h \h//Event Reduction Loop\\
	 \h \h{\bf for} {($j=1$ to $\Event$)} \\
		\h\h\h 	$\mathcal{E} \leftarrow \{\}$;\\
		\h\h\h {\bf for} {$(M \in \mathcal{M})$} \\
		   \h\h\h\h $\mathcal{E}= \mathcal{E} \cup \emph{reduceEvent}(M)$;\\
		\h\h\h $\mathcal{M}$ = All machines in $\mathcal{E}$ that increment $\w(\mathcal{P} \cup
\mathcal{F})$;\\ 
	 \h \h//Minimality Loop \\
	 \h \h$M \leftarrow$ Any machine in $\mathcal{M}$;\\
	 \h \h {\bf while}{ (all states of $M$ have not been combined)} \\
		\h\h\h $\mathcal{C}\leftarrow \emph{reduceState}(M)$;\\
		\h\h\h $M$= Any machine in  $\mathcal{C}$ that increments $\w(\mathcal{P} \cup
\mathcal{F})$; \\
	\h\h $\mathcal{F} \leftarrow \{M\} \bigcup \mathcal{F}$;\\
\h {\bf return} $\mathcal{F}$;
}
\end{minipage}
} % end \fbox
\fbox{
\begin{minipage}[b]  {0.45\linewidth}
{\small
\emph{reduceState}\\
    \h {\bf Input}: Machine $P$ with state set $X_P$, event set $\Sigma_P$\\
		\h and transition function $\alpha_P$;\\ 
    \h {\bf Output}: Largest Machines $< P$ with $\leq |X_P|-1$ states;\\
    \h $\mathcal{B} = \{\}$; \\
    \h {\bf for} $(s_i,s_j \in X_P)$\\ 
	\h\h //combine states $s_i$ and $s_j$\\
	\h\h Set of states, $X_B = X_P$ with $(s_i,s_j)$ combined;\\ 
	\h\h $\mathcal{B} = \mathcal{B} \cup \{$Largest machine consistent with $X_B\}$; \\
    \h {\bf return} largest incomparable machines in $\mathcal{B}$;\\\\
\emph{reduceEvent}\\
    \h {\bf Input}: Machine $P$ with state set $X_P$, event set $\Sigma_P$\\
		\h and transition function $\alpha_P$;\\ 
    \h {\bf Output}: Largest Machines $< P$ with $\leq |\Sigma_P|-1$ events;\\
    \h $\mathcal{B} = \{\}$; \\
    \h {\bf for} $(\sigma \in \Sigma_P)$\\ 
	\h\h Set of states, $X_B = X_P$;\\ 
	\h\h //combine states to self-loop on $\sigma$\\
	\h\h {\bf for} ($s \in X_B$)\\		
		\h\h\h $s = s \cup \alpha_P(s,\sigma)$;\\
	\h\h $\mathcal{B} = \mathcal{B} \cup \{$Largest machine consistent with $X_B\}$; \\
    \h {\bf return} largest incomparable machines in $\mathcal{B}$;\\\\\\\\\\
}\end{minipage}
}
\end{center}
\caption[ ]{Algorithm to generate an \fusion{f}{f} for a given set of primaries $\m{P}$. Note that,
we use the terms \emph{largest}, \emph{incomparable} w.r.t the order defined in section
\ref{secOrder}.}
\label{figFusionAlgo}
\end{figure*}

Given a set of $n$ primaries $\mathcal{P}$, we present an algorithm  to
generate an \fusion{f}{f} $\mathcal{F}$ of $\mathcal{P}$. The number of faults to be corrected, $f$,
is an input parameter based on the system's requirements. The algorithm also takes as input two parameters
$\State$ and $\Event$ and ensures (if possible) that each machine in
$\mathcal{F}$ has at most $(N-\State)$ states and at most $(|\Sigma|-\Event)$ events, where $N$ and
$\Sigma$ are the number of states and events in the $\RCP$. Further,
we show that $\mathcal{F}$ is a minimal fusion of $\mathcal{P}$. The algorithm has time complexity
polynomial in $N$. 

The \emph{genFusion} algorithm executes $f$ iterations and in each iteration adds a machine 
to $\m{F}$ that increases $\w(\m{P \cup F})$ (referred to as $\w$) by one. At
the end of $f$ iterations, $\w$ increases to $f+1$ and hence $\m{P \cup F}$ can
correct $f$ faults. The algorithm ensures that the backup selected in each iteration is optimized for states and
events. In the following paragraphs, we explain the \emph{genFusion} algorithm in detail, followed
by an example to illustrate its working. 

In each iteration of the \emph{genFusion} algorithm (Outer Loop), we first identify the set of weakest edges in
$\m{P \cup F}$ and then find a machine that covers these edges, thereby increasing $\w$
by one. We start with the $\RCP$, since it always increases $\w$. The `State
Reduction Loop' and the `Event Reduction Loop' successively reduce the states and events of the
$\RCP$. Finally the `Minimality Loop' searches as deep into the closed partition set of the
$\RCP$ as possible for a reduced state machine, without explicitly constructing the lattice. 

\emph{State Reduction Loop}: This loop uses the \emph{reduceState} algorithm in Fig. \ref{figFusionAlgo} to
iteratively generate machines with fewer states than the $\RCP$ that increase $\w$ by one. The \emph{reduceState}
algorithm,  takes as input, a machine $P$ and generates a set of machines in which at least two
states of $P$ are combined. For each pair of states $s_i,s_j$ in $X_P$, the \emph{reduceState} algorithm, first
creates a partition of blocks in which $(s_i,s_j)$ are combined and then constructs the largest
machine consistent with this partition. Note that, `largest' is based on the order
specified in section \ref{secOrder}. This procedure is repeated for all pairs in $X_P$ and the
largest incomparable machines among them are returned. At the end of $\State$ iterations of the state reduction loop, we generate a set of machines $\mathcal{M}$ each of which 
increases $\w$ by one and contains at most $(N-\State)$ states, if such  machines exist.  

\emph{Event Reduction Loop}: Starting with the state reduced machines in $\mathcal{M}$, the event
reduction loop uses the \emph{reduceEvent} algorithm in Fig. \ref{figFusionAlgo} to
generate reduced event machines that increase $\w$ by one. The \emph{reduceEvent}
algorithm,  takes as input, a machine $P$ and generates a set of machines that contain at least one
event less than $\Sigma_P$. To generate a machine less than any given
input machine $P$,
that does not contain an event $\sigma$ in its event set, the \emph{reduceEvent} algorithm
combines the states such that they loop
onto themselves on $\sigma$. The algorithm then constructs the largest machine that contains these states in the combined form. This
machine, in effect, ignores $\sigma$. This procedure is repeated for all events in $\Sigma_P$ and the
largest incomparable machines among them are returned. At the end of $\Event$ iterations of the
event reduction loop, we generate a set of machines $\mathcal{M}$ each of which increases $\w$ by
one and contains at most $(N-\State)$ states and at most $(|\Sigma|-\Event)$ events, if such
machines exist.  \footnote{In Appendix \ref{secAppEvReduction}, we present the concept of the
event-based decomposition of machines to replace a given machine $A$ with a set of machines that
contain fewer events than $\Sigma_A$.}

\emph{Minimality Loop}: This loop picks any machine $M$ among the state and event reduced machines
in $\m{M}$ and uses the \emph{reduceState} algorithm iteratively to generate a machine less than $M$ that
increases $\w$ by one until no further state reduction is possible i.e., all the states of $M$ have
been combined. Unlike the state reduction loop (which also uses the \emph{reduceState}
algorithm), in the minimality loop we never exhaustively explore all state reduced machines. After
each iteration of the minimality loop, we only pick \emph{one} machine that increases $\w$ by one. 

Note that, in all three of these inner loops, if in any iteration, no reduction is achieved,
then we simply exit the loop with the machines generated in the previous iteration. We use the example in Fig. \ref{figStEvRedLattice} with $\mathcal{P}=\{A,B,C\},
f=2, \State =1$ and $\Event=1$, to explain the \emph{genFusion} algorithm. Since $f=2$, there are two iterations of the outer loop and in each iteration we generate one machine.
Consider the first iteration of the outer loop. Initially, $\m{F}$ is empty and we need to add
a machine that covers the weakest edges in $G(\{A,B,C\})$. 

To identify the
weakest edges, we need to identify the mapping between the states of the $\RCP$
and the states of the primaries. For example, in Fig. \ref{figStEvRedLattice}, we need to map the states of the $\RCP$ to $A$. The starting states
are always mapped to each other and hence $r^0$ is mapped to $a^0$. Now $r^0$ on event $0$
transitions to $r^2$, while $a^0$ on event $0$ transitions to $a^1$. Hence, $r^2$ is mapped to
$a^1$. Continuing this procedure for all states and events, we obtain the mapping shown, i.e,
$a^0=\{r^0, r^1, r^5, r^6\}$ and $a^1=\{r^2, r^3, r^4, r^7\}$. Following this procedure for all
primaries, we can identify the weakest edges in $G(\{A,B,C\})$ (Fig. \ref{figFaultGraph} $(ii)$). In
Fig. \ref{figStEvRedLattice}, $M_1$, $M_2$ and $F_2$ are some of the largest incomparable machines
that contain at least one state less than the $\RCP$ (the entire set is too large to be enumerated
here). All three of these machines increase $\w$ and at the end of the one and only iteration of the state reduction
loop, $\mathcal{M}$ will contain at least these three machines.

The event reduction loop tries
to find machines with fewer events than the machines in $\m{M}$. For example, to generate a machine less than
$M_2$ that does not contain, say event 2, the \emph{reduceEvent} algorithm combines the blocks of
$M_2$ such
that they do not transition on event 2. Hence, $\{r^0, r^2\}$ in $M_2$ is combined with 
$\{r^4,r^5\}$  and $\{r^1,r^3\}$ is combined with $\{r^6,r^7\}$ to generate machine $F_1$
that does not act on event 2. The only machine less than $M_2$ that does not act on event 1 is
$R_\bot$. Since the \emph{reduceEvent} algorithm returns the largest incomparable machines, only $F_1$ is returned when
$M_2$ is the input. Similarly, with $M_1$ as input, the \emph{reduceEvent} algorithm returns
$\{C,F_1\}$ and with $F_2$ as input it returns $R_\bot$.  Among these machines only $F_1$ increases
$\w$. For example, $C$ does not cover the weakest edge $(r^0,r^1)$ of $G(\m{P})$. 
Hence, at the end of the one and only iteration of the event reduction loop,
$\mathcal{M}=\{F_1\}$.  

As there exists no machine less than
$F_1$, that increases $\w$, at the end of the minimality loop, $M=F_1$. Similarly, in the second
iteration of the outer loop $M=F_2$ and the \emph{genFusion}
algorithm returns $\{F_1,F_2\}$ as the fusion machines that increases $\w$ to three. Hence, using
the \emph{genFusion} algorithm, we have automatically generated the backups $F_1$ and $F_2$ 
shown in Fig. \ref{figMainExample}. Note that, in the worst case, there may exist no efficient
backups and the \emph{genFusion} algorithm will just return a set of $f$ copies of the $\RCP$.
However, our results in section \ref{secEvaluate} indicate that for many examples, efficient backups
do exist. 

\subsection{ Properties of the \emph{genFusion} Algorithm}

In this section, we prove properties of the \emph{genFusion} algorithm with
respect to:  
$(i)$ the number of fusion/backup machines $(ii)$ the number of states in each fusion machine, $(iii)$ the number of events in each fusion
machine and $(iv)$ the minimality of the set of fusion machines $\mathcal{F}$. We first introduce
concepts that are  relevant to the proof of these properties. 

\begin{lemma}\label{lemPrimDist}
Given a set of primary machines $\mathcal{P}$, $\w(\mathcal{P}) = 1$. 
\end{lemma}
\begin{proof}
Given the state of all the primary machines, the state of the $\RCP$ can be uniquely determined.
Hence, there is at least one machine among the primaries that distinguishes between each pair of
states in the $\RCP$ and so, $\w(\m{P}) \geq 1$. In section \ref{secFsmModel}, we assume that the set of machines
in $\m{P}$ cannot correct a single fault and this implies that, $\w(\m{P})\leq 1$. Hence, $\w(\m{P})= 1$. 
\end{proof}

\begin{lemma}\label{lemPrimEdges}
Given a set of primary machines $\mathcal{P}$, let $\mathcal{F'}$ be an \fusion{f}{f} of
$\mathcal{P}$. Each fusion machine $F \in \mathcal{F'}$ has to cover the weakest edges in
$G(\mathcal{P})$.
\end{lemma}
\begin{proof}
From lemma \ref{lemPrimDist}, the weakest edges of $G(\mathcal{P})$ have weight equal to one. Since
$\mathcal{F'}$ is an \fusion{f}{f} of $\mathcal{P}$, $\w(\mathcal{P} \cup \mathcal{F'}) > f$. Also,
each machine in $\mathcal{F'}$ can increase the weight of any edge by at most one. Hence, all the
$f$ machines in $\mathcal{F'}$ have to cover the weakest edges in $G(\m{P})$. 
\end{proof}

%In each iteration of the \emph{genFusion} algorithm (Loop 1), we first construct the fault graph for
%the set of machines in $\mathcal{P} \cup \mathcal{F}$ and identify the weakest edges, i.e edges with
%weight $\w$. We then select
%a machine that covers these edges, thereby increasing $\w$ by one. 
%
Let the weakest edges of $G(\m{P \cup F})$ at the start of the $i^{th}$ iteration of the outer loop
of the \emph{genFusion} algorithm be denoted $E_i$. In the following lemma, we show that the set of weakest edges only increases with each
iteration. 

%In Fig. \ref{figStEvRedLattice}, with $\m{P}=\{A,B,C\}$ edges
%$(r^0,r^1)$, $(r^0,r^7)$ or $(r^0,r^5)$ belong to $E_1$. 
\begin{lemma}\label{lemMonotonic} In the \emph{genFusion} algorithm, for any two iterations $i$ and $j$, if $i<j$,  then $E_i \subseteq E_j$.  
\end{lemma}

\begin{proof} Let the value of $\w$ for the $i^{th}$ iteration be $d$ and the edges with this weight
be $E_i$. Any machine
added to $\m{F}$ can at most increase the weight of each edge by one and it has to
increase the weight of all the edges in $E_i$ by one. So, 
$\w$ for the $(i+1)^{th}$ iteration is $d+1$ and the weight of the edges in $E_i$ will increase
to $d+1$. Hence, $E_i$ will be among the weakest edges in the $(i+1)^{th}$ iteration, or in other
words, $E_i \subseteq E_{i+1}$. This trivially
extends to the result: for any two iterations numbered $i$ and $j$ of the
\emph{genFusion} algorithm, if $i<j$,  then $E_i \subseteq E_j$.  
 \end{proof}

%Since the set of weakest edges only increases with every iteration, we get the following
%observation.
%
%\begin{observation}\label{obEdgeSet} Let $e$ be a weakest edge in some iteration of the
%\emph{genFusion} algorithm.  If there are $k$ machines in $\mathcal{F}$ 
%that cover $e$, then in any valid \fusion{f}{f} of $\mathcal{P}$ there are at least $k$ machines that cover edge $e$.   \end{observation}
%
We now prove one of the main theorems of this paper. 

\begin{theorem}(Fusion Algorithm)\label{thEventState} Given a set of  $n$ machines $\mathcal{P}$, the
\emph{genFusion} algorithm
generates a set of machines $\mathcal{F}$ such that: 
\begin{enumerate}
\item (Correctness) $\mathcal{F}$ is an \fusion{f}{f} of
$\mathcal{P}$. 
\item (State \& Event Efficiency) If each
machine in $\mathcal{F}$ has greater than $(N-\State)$ states and $(|\Sigma|-\Event)$ events, then no \fusion{f}{f} of $\mathcal{P}$ contains a machine with less than or
equal to $(N-\State)$ states and $(|\Sigma|-\Event)$ events. 
%  \item (Event Efficiency) If each
%machine in $\mathcal{F}$ contains more than $|\Sigma|-\Event$ events, then no \fusion{f}{f} of $\mathcal{P}$ contains a machine with less than or
%equal to $|\Sigma|-\Event$ events. 
\item (Minimality) $\mathcal{F}$ is a minimal \fusion{f}{f} of $\mathcal{P}$. 
\end{enumerate}
\end{theorem}
\begin{proof} 
\begin{enumerate}
\item From lemma \ref{lemPrimDist}, $\w(\m{P})=1$. Starting
with the $\RCP$, which always increases $\w$ by one, we add one machine in each iteration  to
$\mathcal{F}$ that
increases by $\w(\mathcal{P} \cup \mathcal{F})$ by one. Hence, at the end of $f$
iterations of the \emph{genFusion} algorithm, we add exactly $f$ machines to $\mathcal{F}$ that
increase $\w$ to $f+1$. Hence, $\mathcal{F}$ is an \fusion{f}{f} of $\mathcal{P}$.  

\item Assume that each machine in $\mathcal{F}$ has greater than $(N-\State)$
states and $(|\Sigma|-\Event)$ events.
 Let there be another \fusion{f}{f} of $\m{P}$ that contains a
machine $F'$ with less than or equal to $(N-\State)$ states and $(|\Sigma|-\Event)$ events. From
lemma \ref{lemPrimEdges}, $F'$ covers the weakest edges in $G(\mathcal{P})$. However, in the first
iteration of the outer loop, the \emph{genFusion} algorithm searches exhaustively for a fusion with less than or equal
to $(N-\State)$ states and $(|\Sigma|-\Event)$ events that covers the weakest edges in $G(\m{P})$.
Hence, if such a machine $F'$ existed, then the algorithm  would have chosen it. %This contradicts the initial assumption.

\item Let there be an \fusion{f}{f} $\mathcal{G} =\{
G_1,..G_f\}$ of $\m{P}$, such that $\mathcal{G}$ is less than
\fusion{f}{f} $\mathcal{F} = \{F_2,F_1,...,F_f\}$. Hence  $\forall j : G_j \leq F_j$. Let $G_i < F_i$ and let $E_i$ be the set of edges that needed to
be covered by $F_i$. It follows from the \emph{genFusion} algorithm, that $G_i$ does not cover at least one edge say $e$ in $E_i$ (otherwise
the algorithm would have returned $G_i$ instead of $F_i$). From lemma
\ref{lemMonotonic}, it follows that if $e$ is covered by $k$ machines
in $\mathcal{F}$, then $e$ has to be covered by  $k$ machines in $\mathcal{G}$. We know that there is a pair of machines  $F_i, G_i$ such that $F_i$
covers $e$ and $G_i$ does not cover $e$. For all other pairs $F_j, G_j$  if $G_j$ covers $e$ then $F_j$ covers $e$ (since $G_j \leq F_j$). Hence $e$
can be covered by no more than $k-1$ machines in $\mathcal{G}$. This implies that $\mathcal{G}$ is
not \fusion{f}{f}. 
\end{enumerate}

\end{proof}

\subsection{Time Complexity of the \emph{genFusion} Algorithm} \label{secFsmTc}
The time complexity of the \emph{genFusion} algorithm is the sum of the time complexities of
the inner loops multiplied by the number of iterations, $f$. We analyze the time complexity of each
of the inner loops. Let the set of machines in $\m{M}$ at the start of the $i^{th}$ iteration of the outer
loop be denoted $\m{M}_i$. 

\emph{State Reduction Loop}: The time complexity of the state
reduction loop for the $i^{th}$ iteration of the outer loop is $T_1+T_2$, where $T_1$ is the time complexity to reduce the states of the
machines in
$\m{M}_i$ and $T_2$ is the time complexity to find  the machines among $\m{S}$ that
increment $\w$. First, let us consider $T_1$. Note that, initially $\m{M}$, i.e, $\m{M}_1$,  contains only the $\RCP$ with $O(N)$ states and for
any iteration of the state reduction loop, each of the machines in $\m{M}_i$ has $O(N)$ states.  Given a machine $M$ with $O(N)$
states, the \emph{reduceState} algorithm generates machines with fewer states than $M$. For each pair of states in
$M$, the time complexity to generate the largest closed partition that contains these states in a
combined block is just $O(N |\Sigma|)$. Since there are $O(N^2)$ pairs of states in $M$, the time
complexity of the \emph{reduceState} algorithm is $O(N^3
|\Sigma|)$. Hence, $T_1 = O(|\m{M}_i|N^3|\Sigma|)$. 

Now, we consider $T_2$. Since, there are $O(N^2)$ pairs of states in each machine in $\m{M}_i$, the
\emph{reduceState} algorithm returns $O(N^2)$ machines. So, $|\m{S}|=O(N^2|\m{M}_i|)$. 
Since there are $O(N^2)$ nodes in the fault graph of $G(\m{P \cup F})$, given
any machine in $\m{S}$, the time complexity to check if it increments $\w$ is $O(N^2)$. Hence,
$T_2=O(|\m{S}|N^2)=O(N^4|\m{M}_i|)$. So, the time complexity of each iteration of the state reduction
loop is $T_1+T_2 = O(|\m{M}_i|N^3|\Sigma| + N^4|\m{M}_i|)$. 

Since the \emph{reduceState} algorithm generates $O(N^2)$ machines per machine in $\m{M}_i$,
$|\m{M}_{i+1}| = N^2|\m{M}_i|$. In
the first iteration $\m{M}$ just contains the $\RCP$ and  $|\m{M}_1| = 1$. Hence, the time
complexity of the state reduction loop is,
$O((N^3|\Sigma| + N^4)(1+N^2+N^4\ldots+ N^{2(\State-1)}))= O((N^3|\Sigma| +
N^4)(\frac{N^{2\State}-1}{N^2-1})$ (the series is a geometric progression). This reduces to 
$O(N^{\State+1}|\Sigma|+N^{\State+2})$. Also, $\m{M}$ contains $O(N^{2\State})$ machines at the
end of the state reduction loop. 

\emph{Event Reduction Loop}: The time complexity analysis for the event reduction loop is similar,
except for the fact that the \emph{reduceEvent} algorithm iterates through $|\Sigma|$ events of the
each machine in $\m{M}$ and returns $O(|\Sigma|)$ machines per machine in $\m{M}$. Also, while the
state reduction loop starts with just one machine in $\m{M}$, the event reduction loop starts with
$O(N^{2\State})$ machines in $\m{M}$.  Hence, the time complexity of each iteration
of the event reduction loop is $O((N|\Sigma|^2 + N^2 |\Sigma|)(N^{2\State})(1+ |\Sigma|+ |\Sigma|^2 \ldots+
|\Sigma|^{\Event-1})) = O((N|\Sigma|^2 + N^2
|\Sigma|)(N^{2\State})(\frac{|\Sigma|^{\Event}-1}{|\Sigma|-1})) =
O(N^{\State+1}|\Sigma|^{\Event+1}+ N^{\State+2}|\Sigma|^{\Event})$. 

\emph{Minimality Loop}: In the minimality loop, we use the \emph{reduceState} algorithm, but only
select one machine per iteration. Also, in each iteration of the minimality loop, the number of
states in $M$ is at least one less than than the number of states in $M$ for the previous iteration. Hence,
the minimality loop executes $O(N)$ iterations with total time complexity, $O((N^3|\Sigma| + N^4)(N))= O(N^4|\Sigma|+N^5)$. 

Since there are $f$ iterations of the outer loop, the time complexity of the \emph{genFusion}
algorithm is, $$O(fN^{\State+1}|\Sigma|+fN^{\State+2}+$$ $$fN^{\State+1}|\Sigma|^{\Event+1}+
fN^{\State+2}|\Sigma|^{\Event}+fN^4|\Sigma|+fN^5)$$ This reduces to, $$O(fN^{\State+1}|\Sigma|^{\Event+1}+
fN^{\State+2}|\Sigma|^{\Event}+fN^4|\Sigma|+fN^5)$$ 

\begin{observation}\label{obGenFusTc} For parameters $\State =0$ and $\Event =0$, the \emph{genFusion} algorithm
generates a minimal \fusion{f}{f} of $\m{P}$ with time complexity $O(fN^4|\Sigma|+fN^5)$, i.e., the time
complexity is polynomial in the number of states of the $\RCP$.
\end{observation}

If there are $n$ primaries each with $O(s)$ states,
then $N$  is $O(s^n)$. Hence, the time complexity of the  \emph{genFusion} algorithm reduces to $O(s^n|\Sigma|f)$. Even though the time complexity of
generating the fusions is exponential in $n$, note that
the fusions have to be generated only once. Further, in Appendix \ref{secAppIncFusion}, we present an incremental approach for the generation of
fusions that improves the time complexity by a factor
of $O(\rho^n)$ for constant values of $\rho$, where $\rho$ is the average state reduction 
achieved by fusion, i.e., (Number of states in the $\RCP$/Average number of states in each fusion
machine).

\section{Detection and Correction of Faults}\label{secFaults}

%In \cite{OgaBhar09}, the time complexity to detect and correct faults is
%$O(n^2  \rho + n \rho  f+N)$, where $n$ is the number of primaries, $f$ is the
%number of crash faults, $s$ is the size of each machine, $N$ is the size of the \RCP{} and $\rho$ is the average state reduction
%achieved by fusion. 
In this
section, we provide algorithms to detect Byzantine faults with time complexity $O(n f)$,
on average, and correct crash/Byzantine
faults with time complexity $O(n  \rho  f)$, with high probability, where $n$ is the number of primaries, $f$ is the
number of crash faults and $\rho$ is the average state reduction
achieved by fusion. Throughout this section, we refer to Fig. \ref{figStEvRedLattice}, with primaries, $\mathcal{P}=\{A,B,C\}$ and backups
 $\mathcal{F} = \{F_1,F_2\}$, that can correct two crash faults. The execution state of the primaries is represented
collectively as a $n$-tuple (referred to as the \emph{primary tuple}) while the state of each
backup/fusion is represented as the set of primary tuples it
corresponds to (referred to as the \emph{tuple-set}). In Fig. \ref{figStEvRedLattice}, if $A$, $B$, $C$ and $F_1$ are in their initial states, then
the primary tuple is $a^0b^0c^0$  and the state of $F_1$ is $f_1^0=\{a^0b^0c^0, a^1b^0c^1,
a^1b^1c^0,a^0b^1c^1\}$ (which corresponds to $\{r^0,r^2,r^4,r^5\}$). 

% We  show that replication incurs almost the same cost
%for the detection and correction of faults. 
%To detect and correct faults, the algorithm in \cite {OgaBhar09}, obtains the list of global states corresponding to the execution states of all
%available $n$ primaries and $f$ fusions,
%and returns the state that appears the maximum number of times among these lists by indexing each of them into an array of size $N$ ($|$\RCP{}$|$).  
%Since the average reduction ratio achieved by fusion is $\rho$, the average size of each such list is $\rho$, giving time complexity 
%$O(n^2  \rho + n \rho  f+N)$ for the algorithm. As $N$ can be exponential in $n$,
%for practical systems 
\begin{figure*}[htb] \begin{center}
\fbox{\begin{minipage}[b]  {0.45\linewidth}
{\small
\emph{detectByz}\\
\h{\bf Input}: set of $f$ fusion states  $B$, primary tuple $r$; \\
\h{\bf Output}: \emph{true} if there is a Byzantine fault and \emph{false} if not;  \\
\h{\bf for} $(b \in B)$ \\
\h	\h{\bf if}  $\neg (\emph{hash\_table}(b) \cdot \emph{contains}(r))$\\
\h		 \h\h {\bf return} \emph{true}; \\
\h {\bf return} \emph{false};\\\\
\emph{correctCrash}\\
\h{\bf Input}: set of available fusion states $B$, primary tuple $r$,\\ \h faults among primaries $t$;\\
\h{\bf Output}: corrected primary $n$-tuple; \\
\h $D \leftarrow \{\}$ //list of tuple-sets\\
\h //find tuples in $b$ within Hamming distance $t$ of $r$ \\
\h {\bf for} $(b \in B)$\\
\h	 \h $S \leftarrow \emph{lsh\_tables}(b) \cdot  \emph{search}(r,t)$;\\
 \h	 \h $D \cdot \emph{add}(S)$; \\
\h {\bf return} Intersection of sets in $D$;%// singleton w.h.p  
}

\end{minipage}
} % end \xbox
\fbox{\begin{minipage}[b]  {0.45\linewidth}
{\small
\emph{correctByz}\\
\h{\bf Input}: set of $f$ fusion states $B$, primary tuple $r$; \\
\h{\bf Output}: corrected primary $n$-tuple; \\
\h $D \leftarrow \{\}$ //list of tuple-sets\\
\h  //find tuples in $b$ within Hamming distance  $\lfloor f/2 \rfloor$ of $r$ \\
\h {\bf for} $(b \in B)$\\ 
\h	\h $S \leftarrow \emph{lsh\_tables}(b) \cdot  \emph{search}(r,\lfloor f/2 \rfloor)$;\\
\h	\h $D \cdot \emph{add}(S)$; \\
\h $G \leftarrow$ Set of tuples that appear in $D$;\\
%\h //$G$ is just for efficient indexing\\
\h $\emph{V} \leftarrow$ Vote array of size $|G|$;\\
\h{\bf for} $(g \in G)$ \\
\h		\h // get votes from fusions  \\
\h		\h $V[g] \leftarrow$ Number of times $g$ appears in $D$; \\
\h		\h // get votes from primaries \\
\h		\h{\bf for} {$(i=1$ to $n)$} \\
\h			\h\h{\bf if}$(r[i] \in g)$ \\
\h				\h\h\h $V[g]++$; \\
\h{\bf return} Tuple $g$ such that $V[g] \geq n+\lfloor f/2 \rfloor$; 
}

\end{minipage}
} % end \fbox
\end{center}
\caption[ ]{Detection and correction of faults.}
\label{figFaultAlgo}
\end{figure*}

\subsection{Detection of Byzantine Faults}\label{secByzDet} 
Given the primary tuple and the tuple-sets corresponding to the fusion states, the
\emph{detectByz} algorithm in Fig. \ref{figFaultAlgo} detects up to $f$ Byzantine faults (liars).
Assuming that the tuple-set of each fusion state is stored in a permanent hash table
at the recovery agent, the \emph{detectByz} algorithm simply checks if the primary tuple $r$ is
present in each backup tuple-set $b$. In Fig. \ref{figStEvRedLattice}, if the states of machines
$A$, $B$, $C$, $F_1$
and $F_2$ are $a^1$, $b^1$, $c^0$, $f_1^1$ and $f_2^1$ respectively, then the algorithm flags a
Byzantine fault, since $a^1b^1c^0$ is not present in either $f_1^1=\{a^0b^1c^0, a^1b^1c^1,
a^0b^0c^1,a^1b^0c^0\}$  or $f_2^1=\{a^0b^1c^0, a^1b^0c^1\}$. 

To show that $r$  is not  present in at least one of the
backup tuple-sets in $B$ when there are liars, we make two observations. First, we are only concerned
about machines that lie within their state set.
For example, in Fig. \ref{figStEvRedLattice}, suppose the true state of $F_2$ is $f_2^0$. To lie,  
if $F_2$ says it state is any number apart from $f_2^1$, $f_2^2$ and $f_2^3$, then that can
be detected easily. 

Second, like the fusion states, each primary state can be expressed as
a tuple-set that contains the $\RCP$ states it belongs to. Immaterial of whether $r$ is 
correct or incorrect (with liars), it will be present in all the truthful primary
states. For example, in Fig. \ref{figStEvRedLattice}, if the correct primary tuple is $a^0b^0c^0$ 
then $a^0=\{a^0b^0c^0,a^0b^1c^0,a^0b^1c^1,a^0b^0c^1\}$ contains $a^0b^0c^0$. If $B$ lies, then the
primary tuple will be $a^0b^1c^0$, which is incorrect. Clearly, $a^0$ contains this incorrect primary tuple as well.  

%In this example, $A$ and
%$B$ were lying about their state. If they had truthfully reported their state as $a^0$ and $b^1$,
%then the primary tuple $\{a^0b^1c^0\}$ would have been found in the backup states. 
%
\begin{theorem}\label{thByzDet} Given a set of $n$ machines $\mathcal{P}$ and an \fusion{f}{f} $\mathcal{F}$
corresponding to it, the \emph{detectByz} algorithm detects up to $f$ Byzantine faults among them.  
\end{theorem} 
\begin{proof}
Let $r$ be the correct primary tuple. Each primary tuple is present in exactly one
fusion state (the fusion states partition the $\RCP$ states), i.e, the correct fusion
state. Hence, the incorrect fusion states (liars) will not contain $r$ and the fault will be detected.  
If $r$ is incorrect (with liars), then for the fault to go undetected, $r$ must be present
in all the fusion states. 

If $r^{c}$ is the correct primary tuple, then the truthful fusion
states have to contain $r^{c}$ as well, which implies that they contain
$\{r,r^{c}\}$ in the same tuple-set. As observed above, the truthful 
primaries will also contain $\{r,r^{c}\}$ in the same tuple-set. So the execution state of all
the truthful machines contain $\{r,r^{c}\}$ in the same tuple-set. Hence less than or equal to $f$ machines,
i.e, the liars, can contain $r$ and $r^{c}$ in distinct tuple-sets. This contradicts the fact that $\mathcal{F}$ is a
\fusion{f}{f} with greater than $f$ machines separating each pair of $\RCP$ states. 
\end{proof}

We consider the space complexity for maintaining the hash tables at the \recover{}. Note that, the
space complexity to maintain a hash table is simply the number of points in the hash table
multiplied by the size of each point. In our solution we hash the tuples belonging to the fusion
states. In each fusion machine, there are $N$ such tuples, since the fusion states partition the
states of the $\RCP$. 
Each tuple contains $n$ primary states each of size $\log s$, where $s$ is the maximum number of
states in any primary. For example, $a^0b^1c^0$ in $f_1^1$ contains three primary states $(n=3)$ and since
there are two states in $A$ $(s=2)$ we need just one bit to represent it. Since there are $f$
fusion machines, we hash a total of $Nf$ points, each of size $O(n\log s)$. Hence, the space complexity
at the \recover{} is $O(Nfn\log s)$. 

Since each fusion state is maintained as a hash table, it will take $O(n)$ time (on average) to check
if a primary tuple with $n$ primary states is present in the fusion state.  Since there are $f$
fusion states, the time complexity for the \emph{detectByz}
algorithm  is $O(n  f)$ on average. Even for replication,
the \recover{} needs to compare the state of $n$ primaries with the state of each of its $f$
copies, with time complexity $O(n  f)$. In terms of message complexity, in fusion, we need to
acquire the state of $n+f$ machines to detect the faults, while for replication, we need to acquire
the state of $2nf$ machines. 

\subsection{Correction of Faults}
Given a primary tuple $r$ and the tuple-set of a fusion state, say $b$, consider the problem of finding the tuples in $b$ 
that are within Hamming distance $f$ of $r$. This is the key concept that we use for the correction
of faults, as explained in sections \ref{secCrashCorr} and \ref{secByzCorr}.  
In Fig. \ref{figStEvRedLattice}, the tuples in $f_1^0=\{a^0b^0c^0, a^1b^0c^1,
a^1b^1c^0,a^0b^1c^1\}$ that are within Hamming distance one of a
primary tuple $a^0b^0c^1$ are  $a^0b^0c^0$, $a^1b^0c^1$ and $a^0b^1c^1$. An efficient solution to finding
the points among a large set within a certain Hamming distance of a query point is \emph{locality sensitive hashing} (LSH)
\cite{andoniIndyk06,Gionis97similaritysearch}. Based on this, we first select $L$ hash
 functions $\{g_1 \ldots g_L\}$ and for each $g_j$ we associate an ordered set (increasing order) of $k$ numbers
$C_j$ picked uniformly at random from $\{0 \ldots n\}$. The hash function $g_j$ takes as input an
$n$-tuple, selects the coordinates from them as specified by the numbers in $C_j$ and returns  the
concatenated bit representation of these coordinates. At the \recover{}, for each fusion state we
maintain $L$ hash tables, with the functions selected above, and hash each tuple in the fusion state. 
In Fig. \ref{figLSH} $(i)$, $g_1$ and $g_2$ are associated
with the sets $C_1=\{0,1\}$ and $C_2=\{0,2\}$ respectively.  
Hence, the tuple $a^1b^0c^1$ of $f_1^0$, is hashed into the $2^{nd}$ bucket of $g_1$ and the $3^{rd}$
bucket of $g_2$. 

\begin{figure*}[htb] 
\centerline{ 
\scalebox{0.40}{ 
\includegraphics{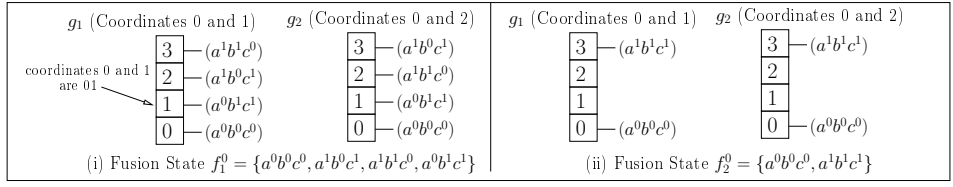} } } 
\caption{LSH example for fusion states in Fig. \ref{figStEvRedLattice} with $k=2$, $L=2$.} \label{figLSH} 
\end{figure*} 

Given a primary tuple $r$ and a fusion state $b$, to find the tuples among $b$ that are within a
Hamming distance $f$ of $r$, we obtain
the points found in the buckets $g_j(r)$ for $j = 1 \ldots L$ maintained for $b$ and return those that are within
distance of $f$ from $r$. In Fig. \ref{figLSH} $(i)$, let $r=  a^0b^1c^0$, $f=2$, and $b =f_1^0$.
The primary tuple $r$ hashes
into the $1^{st}$ bucket of $g_1$ and the $0^{th}$ bucket of $g_2$ which contains the points
$a^0b^1c^1$ and $a^0b^0c^0$ respectively. Since both of them are withing Hamming distance two of
$r$, both the points are returned. If we set $L = \log_{1-\gamma^k}\delta$, where
$\gamma=1-f/n$, such that $(1-\gamma^k)^L \leq \delta$, then any $f$-neighbor of a point $q$
is returned with probability at least $1-\delta$ \cite{andoniIndyk06,Gionis97similaritysearch}. 
In the following sections, we present algorithms for the correction of crash and Byzantine faults
based on these LSH functions. 

\subsubsection{Crash Correction}\label{secCrashCorr} Given the primary tuple (with
possible gaps due to faults) and the tuple-sets of the available
fusion states, the \emph{correctCrash} algorithm in Fig. \ref{figFaultAlgo} corrects up to $f$
crash faults. The algorithm finds the set of tuple-sets $S$ in each fusion state $b$, where each
tuple belonging to $S$ is
within a Hamming distance $t$ of the primary tuple $r$. Here, $t$ is the number of faults among the
primaries. To do this efficiently, we use the LSH
tables of each fusion state.  The set $S$ returned for each fusion state is stored in a list $D$. 
If the intersection of the sets in $D$ is singleton, then we return
that as the correct primary tuple. If the intersection is empty, we need to exhaustively
search each fusion state for points within distance $t$ of $r$ (LSH has not returned
all of them), but this happens with a very low probability
\cite{andoniIndyk06,Gionis97similaritysearch}. %We show in theorem \ref{thmCrashCorr} that the
%intersection of the sets in $D$ is always singleton. 

In Fig. \ref{figStEvRedLattice}, assume crash faults in $B$ and $C$.   
Given the states of $A$, $F_1$
and $F_2$ as $a^0$, $f_1^0$ and $f_2^0$ respectively, the tuples within Hamming distance two of $r
=a^0.\{empty\}.\{empty\}$ among states $f_1^0=\{a^0b^0c^0, a^1b^0c^1, a^1b^1c^0,a^0b^1c^1\}$ and  $f_2^0 =\{a^0b^0c^0, a^1b^1c^1\}$
 are $\{a^0b^0c^0, a^0b^1c^1\}$ and $\{a^0b^0c^0\}$ respectively. The algorithm returns
their intersection, $a^0b^0c^0$ as the corrected primary tuple. In the following theorem, we prove
that the \emph{correctCrash} algorithm returns a unique primary tuple. 

\begin{theorem}\label{thmCrashCorr}
Given a set of $n$ machines $\mathcal{P}$ and an \fusion{f}{f} $\mathcal{F}$
corresponding to it, the \emph{correctCrash} algorithm corrects up to $f$ crash faults among them.  
\end{theorem} 
\begin{proof} 
Since there are $t$ gaps due to $t$ faults in the primary tuple $r$, the tuples among the backup tuple-sets within a Hamming distance $t$ of $r$, are
the tuples that contain $r$ (definition of Hamming distance). Let us assume that the intersection of
the tuple-sets among the fusion states containing $r$ is not singleton. Hence all the available fusion states
have at least two $\RCP$ states, $\{r^i,r^j\}$, that contain $r$. Similar to the proof in theorem
\ref{thByzDet}, since both $r^i$ and
$r^j$ contain $r$, these states will be present in the same tuple-sets of all the available primaries 
as well.  Hence less than or equal to $f$ machines,
i.e, the failed machines, can contain $r^i$ and $r^j$ in distinct tuple-sets. This contradicts the
fact that $\mathcal{F}$ is an
\fusion{f}{f} with greater than $f$ machines separating each pair of $\RCP$ states. 
\end{proof}

The space complexity analysis is similar to that for Byzantine detection since we maintain hash
tables for each fusion state and hash all the tuples belonging to them. Assuming $L$ is a constant,
the space complexity of storage at the \recover{} is $O(Nfn\log s)$. 

Let $\rho$ be the average state reduction achieved by our fusion-based technique. Each fusion
machine partitions the states of the $\RCP$ and the average size of
each fusion machine is $N/\rho$. Hence, the number of tuples (or points) in each fusion state is
$\rho$. This implies that there can be $O(\rho)$ tuples in each fusion state that are within
distance $f$ of $r$. So, the cost of hashing $r$ and retrieving $O(\rho)$ $n$-dimensional points from $O(f)$
fusion states in $B$ is $O(n \rho f)$ w.h.p (assuming $k,L$ for the LSH tables are constants). So,
the cost of generating $D$ is $O(n \rho f)$ w.h.p. Also, the number of tuple sets in $D$ is $O(\rho f)$. 

In order to find the intersection of the
tuple-sets in $D$ in linear time, we can hash the elements of the
smallest tuple-set and check if the elements of the other tuple-sets are part of this set. 
The time complexity to find the intersection among the
$O(\rho f)$ points in $D$, each of size $n$ is simply $O(n \rho f)$. Hence, the overall time
complexity of the \emph{correctCrash} algorithm  is $O(n \rho f)$ w.h.p. Crash correction in replication
involves copying the state of  the copies of the $f$ failed primaries which has time complexity
$\theta(f)$. In terms of message complexity, in fusion, we need to acquire the state of all $n$ machines
that remain after $f$ faults. In replication we just need to acquire the copies of the $f$ failed
primaries. 

\subsubsection{Byzantine Correction}\label{secByzCorr}

Given the primary tuple  and the tuple-sets of the 
fusion states, the \emph{correctByz} algorithm in Fig. \ref{figFaultAlgo} corrects up to $\myfloor{f/2}$ Byzantine 
faults. The algorithm finds the set of
tuples among the tuple-sets of each fusion state that are within Hamming distance
$\lfloor f/2 \rfloor$ of the primary tuple $r$ using the LSH tables and stores them in list $D$. It then constructs a vote vector $V$ for each unique tuple in this list. 
The votes for each tuple $g \in V$ is the number of times it appears in $D$
plus the number of primary states of $r$ that appear in $g$. 
The tuple with greater than or equal to $n+\lfloor f/2 \rfloor$ votes is the correct primary tuple.
When there is no such tuple, we need to exhaustively search each fusion state for points within
distance $\lfloor f/2 \rfloor$ of $r$ (LSH has not returned all of them). 
In Fig. \ref{figStEvRedLattice}, let the states of machines $A$, $B$, $C$ $F_1$
and $F_2$ are $a^0$, $b^1$, $c^0$, $f_1^0$ and $f_2^0$ respectively, with one liar among them
$(\lfloor f/2 \rfloor = 1)$. The tuples within Hamming
distance one of $r=a^0b^1c^0$ among $f_1^0=\{a^0b^0c^0, a^1b^0c^1, a^1b^1c^0,a^0b^1c^1\}$ and
$f_2^0=\{a^0b^0c^0, a^1b^1c^1\}$ are $\{a^0b^0c^0, a^1b^1c^0,a^0b^1c^1\}$ and $\{a^0b^0c^0\}$ respectively.
Here, tuple $a^0b^0c^0$ wins a vote each from $F_1$ and $F_2$ since $a^0b^0c^0$ is present in $f_1^0$ and
$f_2^0$.  It also wins a vote each from $A$ and $C$, since the current states of $A$ and $C$, $a^0$
and $c^0$, are present in $a^0b^0c^0$. 
The algorithm returns $a^0b^0c^0$ as the true primary tuple, since $n+\lfloor f/2 \rfloor= 3+1 =4$. 
We show in the following theorem that the true primary
tuple will always get sufficient votes.

\begin{theorem}Given a set of $n$ machines $\mathcal{P}$ and an \fusion{f}{f} $\mathcal{F}$
corresponding to it, the \emph{correctByz} algorithm corrects up to $\myfloor{f/2}$ Byzantine faults among them.  
\end{theorem} 
\begin{proof}
We prove that the true primary tuple, $r^{c}$ will uniquely get  greater than or equal to
$(n+\myfloor{f/2})$ votes. Since
there are less than or equal to $\myfloor{f/2}$ liars, $r^{c}$ will be present in the 
tuple-sets of greater than or equal to $n+\myfloor{f/2}$ 
machines. Hence the number of votes to $r^{c}$, $V[r^{c}]$ is greater than or equal to
$(n+\myfloor{f/2})$. An incorrect
primary tuple $r^{w}$ can get votes from less than or equal to $\myfloor{f/2}$ machines (i.e, the liars) and the truthful machines
that contain both $r^{c}$ and $r^{w}$ in the same tuple-set. Since $\m{F}$ is an
\fusion{f}{f} of $\m{P}$, among all the $n+f$ machines, less than $n$ of them contain
$\{r^{c},r^{w}\}$ in the same tuple-set. Hence, the number of votes to $r^{w}$,
$V[r^{w}]$ is less than $(n +\myfloor{f/2})$ which is less than $V[r^{c}]$. 
\end{proof}

The space complexity analysis is similar to crash correction.
The time complexity to generate $D$, same as that for crash fault correction is $O(n \rho f)$ w.h.p. If we
maintain $G$ as a hash table (standard hash functions), to obtain votes from the fusions, we just
need to iterate through the $f$ sets in $D$, each containing $O(\rho)$ points of size $n$ each and
check for their presence in $G$ in constant time. Hence the time complexity to obtain votes from the
backups is $O(n \rho f)$. Since the size of $G$ is $O(\rho f)$, the time complexity to obtain
votes from the primaries is again $O(n \rho f)$, giving over all time complexity $O(n \rho f)$
w.h.p. In the case of replication, we just need to obtain the majority across $f$ copies of each
primary with time complexity $O(nf)$. The message complexity analysis is the same as Byzantine
detection, because correction can take place only after acquiring the state of all machines and
detecting the fault. 
 
\section{Practical use of Fusion in the MapReduce Framework} \label{secMapReduce}
To motivate the practical use of fusion, we discuss its
potential application to the MapReduce framework which is used  to model large scale
distributed computations. Typically, the  MapReduce framework is built using the master-worker configuration where
the master assigns the map and reduce tasks to various workers. While the map tasks perform the actual
computation on the data files received by it as $<$key, value$>$ pairs, the reducer tasks aggregate the
results according to the keys and writes it to the output file. 

\col{Note that, in batch processing
application for MapReduce, fault tolerance is based on passive replication. So, a task that failed
would simply be restarted on another worker node. However, our work is targetted towards applications
such as distributed stream processing, with strict deadlines. Here, active replication is often used for fault tolerance
\cite{Shah2004,balazinska2005fault}. Hence, tasks are replicated at the beginning of the
computation, to ensure that despite failures there are sufficient workers remaining.} 

In this
paper, we
focus on the \emph{distributed grep} application based on the MapReduce framework. Given a continuous
stream of data files, the grep application checks if every line of the file matches patterns defined by 
regular expressions (modeled as \FSMs{}). Specifically, we assume that the expressions are $((0 +
1)(0 + 1))$*, $((0+2)(0+2))$* and $(00)$* 
modeled by $A$, $B$, $C$ shown in Fig. \ref{figMainExample}. We show using a simple case study that the current replication
based solution requires 1.8 million map tasks while our solution that combines fusion with
replication requires only 1.4 million map tasks. This results in considerable savings in space and other
computational resources. 
%which has a wide range of  
%applications including machine learning, data mining and syntactic analysis. G
%In the original MapReduce paper \cite{Dean2008} 
%handling faults among the mappers is primarily based on checkpointing  in which the processes periodically write to 
%permanent storage. This approach is very simple, but increases latency during fault-free operation, since writing to disk is
%costly. For many applications, such as stream processing, this is inadequate. These applications use
%active replication in which the data tuples are sent to identical copies of the mappers. In
%case of failure, the state of the mapper can be recovered from the corresponding replica. In this
%paper, we present a solution that combines fusion-with replication that guarantees the same fault
%tolerance as a pure replication based solution, but with far fewer backups. We do so without
%compromising on most deadlines.    
%
\begin{figure}[htb] 
\centerline{ 
\scalebox{0.40}{ 
\includegraphics{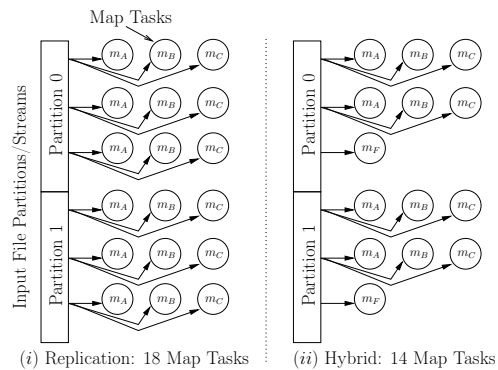} } } 
\caption{Replication vs. Fusion for \emph{grep} using the MapReduce framework.}
\label{figMapReduce} 
\end{figure} 
\subsection{Existing Replication-based Solution}

We first outline a simplified version of a pure replication based solution to correct two crash
faults in Fig. \ref{figMapReduce} $(i)$. Given an input file
stream, the
master splits the file into smaller partitions (or streams) and breaks these partitions into $<$file
name, file content$>$ tuples. For each partition, we maintain three primary map tasks $m_A$, $m_B$
and $m_C$ that output the lines that match the regular expressions modeled by $A$, $B$ and $C$ respectively. To
correct two crash faults, we
maintain two additional copies of each primary map task for every partition. The master sends tuples belonging to
each partition to the primaries and the copies. The reduce phase just
collects all lines from these map task and passes them to the user. Note that, the reducer receives inputs from the
primaries and its copies and simply discards duplicate inputs.
Hence, the copies help in both fault tolerance and load-balancing. 

When map tasks fail, the state
of the failed  tasks can be recovered from one of the remaining
copies. From Fig. \ref{figMapReduce} $(ii)$, it is clear that each file partition 
requires nine map tasks. In such systems, typically, the input files are large enough to be
partitioned into 200,000 partitions \cite{Dean2008}. Hence, replication requires 1.8 million
map tasks. 
%During normal operation the primaries and its copies send heart
%beats to each other for crash fault detection. 
\subsection{Hybrid Fusion-based Solution}

In this section, we outline an alternate solution based on a combination of replication and fusion, 
as shown in Fig. \ref{figMapReduce} $(ii)$. For each partition, we maintain just one additional copy of each primary and
also maintain one fused map task, denoted $m_F$ for the entire set of
primaries. The fused map task searches for the regular expression $(11)$* modeled by $F_1$ in Fig.
\ref{figMainExample}. Clearly, this solution can correct two crash faults among the primary
map tasks, identical to the replication-based solution. The reducer operation
remains identical. The
output of the fused map task is relevant only for fault tolerance and hence it does not send its
output to the reducer. Note that since there is only one additional copy of each primary, we compromise on
the load balancing as compared to pure replication. However, we require only seven map tasks as
compared to the nine map tasks required by pure replication. 
%During regular operation, the copies send
%heart beats to each other as well as the fused copy to enable fault detection. 

When only one fault occurs among the map tasks, the state of the failed map task 
can be recovered from the remaining copy with very little overhead. Similarly, if two faults occur \emph{across} the
primary map tasks, i.e., $m_A$ and $m_B$ fail, then their state can be recovered from the remaining copies. 
Only in the relatively rare event that two faults occur among the
copies of the same primary, that the fused map task has to be used for recovery.  For example, if both
copies of $m_A$ fail, then $m_F$ needs to acquire the state of $m_B$ and $m_C$ (any of the
copies) and perform the algorithm for crash correction in \ref{secCrashCorr} to recover the state of
$m_A$. Considering 200,000 partitions, the hybrid approach needs only 1.4 million map
tasks which is 22\% lesser map tasks than replication, even for this simple example. Note that as $n$
increases, the savings in the number of map tasks increases even further. This results in considerable
savings in terms of $(i)$ the state space required by these map tasks $(ii)$ resources such as the power consumed by
them. 
 
%{\small
%%\begin{tabular}{|p{0.95in}|p{0.45in}|p{0.45in}|p{0.45in}|p{0.45in}|}
%\begin{tabular}{|c|c|c|c|c|}
%\hline
%Machines& $|\RCP{}|$& RCP Events &  Fusion Event & State Space saved \emph{genFusion}\\
%\hline
%beecount, bbara, lioni	&490&	16	&4	&99.71\\
%beecount, lion , lioni	&196&	8	&4	&99.73\\
%beecount, lioni, bbtas	&294&	8	&4	&98.15\\
%beecount, lioni, modulo1&588&	8	&1	&97.96\\
%beecount, lioni, lion	&196&	8	&4	&98.24\\
%beecount, lioni, mc	&196&	8	&7	&99.06\\
%beecount, lioni, shiftreg&392&	8	&2	&97.96\\
%bbara 	, lion , lion	&160&	16	&4	&98.44\\
%\hline
%\end{tabular}
%}
%\end{table}
%%We implemented the algorithms in Fig. \ref{figFaultAlgo} and performed three major
%experiments. Firstly, our focus was on finding the state space achieved by fusion for these real
%world machines as compared to a simple replication-based approach. Secondly, we measure the savings
%in time achieved by the incremental fusion algorithm \emph{incFusion} over the \emph{genFusion}
%algorithm. Lastly, we measured the event-reduction achieved by our fusion algorithm. 
%

%\section{Comparative Study: Replication vs. Fusion}

\section{Experimental Evaluation}\label{secEvaluate}

\begin{table}[ht]\centering
\caption{MCNC' 91 Benchmark Machines}
{\small
\begin{tabular}{|c|c|c|}
\hline
\emph{Machines} & \emph{States} & \emph{Events}\\
\hline
dk15 & 4 & 8  \\
\hline
bbara & 10 & 16  \\
\hline
mc & 4  & 8  \\
\hline 
lion & 4  & 4  \\
\hline
bbtas &  6  & 4   \\
\hline
tav & 4  & 16   \\
\hline
modulo12 & 12  & 2  \\
\hline
beecount & 7 & 8  \\
\hline
shiftreg & 8  & 2  \\
\hline
\end{tabular}\label{tabMCNC}
}
\end{table}

%In \cite{OgaBhar09}, we evaluate fusion for simple examples such
%as counters and dividers. 
In this section, we evaluate fusion using the MCNC'91 benchmarks
\cite{Yang91logicsynthesis} for \FSMs{}, widely used for 
research in the fields of logic synthesis and finite state machine synthesis 
\cite{Mishchenko06dag,YouraIMF98}. In Table \ref{tabMCNC}, we specify the number of states and
number of events/inputs for the benchmark machines presented in our results.  We implemented an incremental version of the \emph{genFusion}
algorithm (Appendix \ref{secAppIncFusion}) in Java 1.6 and 
compared the performance of fusion with replication for 100 different combinations of the benchmark
machines, with $n=3$, $f=2$, $\Event=3$ and present some of the results in Table \ref{tabResults}. The 
implementation with detailed results are available in \cite{mapleFusionFSM}. 

Let the primaries be denoted $P_1$, $P_2$ and $P_3$ and the fused-backups $F_1$ and $F_2$. Column 1 
of Table \ref{tabResults} specifies the names of three primary \FSMs{}. Column 2 specifies the backup space
required for replication ($\prod_{i=1}^{1=3}|P_i|^f$)
, column 3 specifies the
backup space for fusion $(\prod_{i=1}^{i=2}|F_i|$) and column 4 specifies the
percentage state space savings ((column 2-column 3)* 100/column 2). Column 5 specifies the total number of
primary events, column 6 specifies the average number of events across $F_1$ and $F_2$ and the last column
specifies the percentage reduction in events ((column 5-column 6)*100/column 5). 

\col{For example, consider the first row of Table \ref{tabResults}. The primary machines are the ones
named dk15, bbara and mc. Since the machines have 4, 10 and 4 states respectively (Table
\ref{tabMCNC}), the replication
state space for $f =2$, is the state space for two additional copies of each of these machines, which is
$(4* 10*4)^2$ = 25600. The two fusion machines generated for this set of primary machines each had
140 states and hence, the total state space for fusion as a solution is 19600. For the benchmark machines, the events are
binary inputs. For example, as seen in Table \ref{tabMCNC}, dk15 contains eight events. Hence, the event set of dk15 = $\{0,1, \ldots,7\}$. The event sets of the primaries is
the union of the event set of each primary. So, for the first row of Table
\ref{tabResults}, the primary event set is $\{0,1,\ldots 15\}$. In this example, both fusion
machines had 10 events and hence, the average number of fusion events is 10.}

%\removespace{0.10}
\begin{table*}[ht]\centering
\caption{Evaluation of Fusion on the MCNC'91 Benchmarks}
%\begin{tabular}{|p{0.80in}|p{0.35in}|p{0.90in}|p{0.90in}|}
{\small
\begin{tabular}{|c|p{0.70in}|p{0.70in}|p{0.70in}||p{0.50in}|p{0.45in}|p{0.70in}|}
\hline
\emph{Machines}& \emph{Replication State Space}& \emph{Fusion State Space} & \emph{ \%
Savings State Space} & \emph{Primary Events} &
\emph{Fusion Events}&  \emph{\% Reduction Events} \\
\hline
dk15, bbara,  mc	&25600	&19600	&23.44	&16		&10  &         37.5\\
\hline
lion, bbtas, mc	 	&9216	&8464	&8.16	&8		&7   &        12.5\\
\hline
lion, tav,  modulo12	&36864	&9216	&75	&16		&16  &   	 0\\
\hline
lion, bbara, mc	 	&25600	&25600	&0	&16		&9   &   	43.75\\
\hline
tav, beecount, lion	&12544	&10816	&13.78	&16		&16  &        0\\
\hline
mc, bbtas,  shiftreg	&36864	&26896	&27.04	&8		&7   &        12.5\\
\hline
tav, bbara, mc	 	&25600	&25600	&0	&16		&16  &   	0\\
\hline
dk15, modulo12, mc	&36864	&28224	&23.44	&8		&8   &   	0\\
\hline
modulo12, lion, mc	&36864	&36864	&0	&8		&7   &   	12.5\\
\hline
\end{tabular}\label{tabResults}

}
	\end{table*}
%\removespace{0.09}
The average state space savings in fusion (over replication) is 38\% (with range 0-99\%) over the 100
combination of benchmark machines, while the average event-reduction is 4\% (with range 0-45\%). We also
present results in \cite{mapleFusionFSM} that show that the average savings in time by the incremental approach for generating the fusions (over
the non-incremental approach) is 8\%. 
Hence, fusion achieves significant
savings in space for standard benchmarks, while the event-reduction indicates that for many
cases, the backups will not contain a large number of events. 

\section{Discussion: Backups Outside the Closed Partition Set}\label{secMacOutLat} 

So far in this paper, we have only considered machines that belong to the closed partition set.
In other words, given a set of primaries $\mathcal{P}$, our search for backup machines was
restricted to those that are less than the $RCP$ of $\m{P}$, denoted by $R$. However, it is possible
that efficient backup machines exist \emph{outside} the lattice, i.e., among machines that are not less
than or equal to $R$. In this section, we present a
technique to detect if a machine outside the closed partition set of $R$ can correct faults among the primaries. 
Given a set of machines in $\mathcal{F}$ each less than or equal to $R$, we can determine if
$\m{P \cup F}$ can correct faults based on the $\w$ of $\m{P\cup F}$ (section
\ref{secGraphHamming}). To find $\w$, we first determine the mapping between the states of 
$R$ to the states of each of the machines in $\m{F}$. However, given a set of
machines in $\m{G}$ that are not less than or equal to $R$, how do we generate this mapping?

%\begin{figure}[htb] 
%\fbox{\begin{minipage}[b]  {0.8\linewidth}
%{\small
%\emph{canCorrectFaults}\\
%\h{\bf Input}: Primaries $\mathcal{P}$,  Set of machines $\mathcal{G}$, faults $f$;\\
%\h{\bf Output}:  \emph{true} or \emph{false}\\
%\h$R \leftarrow \RCP(\{\m{P}\})$; \\
%\h$B \leftarrow \RCP(\{R \cup \m{G}\})$; \\
%\h Obtain mapping $m_{B,R}$ from $X_B$ to $X_R$;\\
%\h{\bf for} {$(G \in \m{G})$ }\\
%	\h\h Obtain mapping $m_{B,G}$ from $X_B$ to $X_G$;\\
%	\h\h{\bf for} {$(s \in X_B)$}\\
%		\h\h\h $m_{R,G}(m_{B,G}(s)) =m_{B,R}(s)$\\
%\h/* Now even for $G$ in $\m{G}$, we can determine if \\
%\h\h any $r^i,r^j$ of $X_\RCP$ are in separate blocks of $X_G$*/\\
%\h {\bf if} (greater than $f$ machines among $\m{P} \cup \m{F}$ \\
%\h\h  separate each $r^i, r^j \in X_R$) \\
%	\h\h {\bf return} \emph{true};\\
%\h{\bf else}\\
%	\h\h {\bf return} \emph{false};
%}
%\end{minipage}
%} % end \fbox
%\caption[ ]{Machines Outside the Closed Partition Lattice}
%\label{figMacOutLat}
%\end{figure}
%
%\begin{theorem} Given a set of primaries $\mathcal{P}$, consider a set of machines  
%$\mathcal{G}$ such that $\forall G \in \m{G}, (G || \RCP) \vee
%(G > \RCP)$. There exists an algorithm to determine if $\m{G}$ can correct  $f$
%crash or $\lfloor f/2 \rfloor$ Byzantine faults among the machines in $\m{P}$.    
%\end{theorem}
%\begin{proof}
%\end{proof}
%In Fig. \ref{figMacOutLat}, we present  an algorithm that checks if a given set of machines outside
%the closed partition set of $\RCP({\m{P}})$ can correct faults among the primaries.

To determine the mapping between the states of $R$ to the states of the machines in $\m{G}$, we first generate the $\RCP$ of ${\{R\} \cup \m{G}}$, denoted $B$, 
 which is be greater than all the machines in ${\{R\} \cup \m{G}}$.  Hence, we can determine the
mapping between the states of $B$ and the states of all the machines in ${\{R\} \cup \m{G}}$. Given
this mapping, we can determine the (non-unique) mapping between the states of $R$ and the states
of the machines in $\m{G}$. This enables us to determine $\w(R, {\{R\} \cup \m{G}})$. If this $\w$ is 
greater than $f$, then $\m{G}$ can correct  $f$
crash or $\lfloor f/2 \rfloor$ Byzantine faults among the machines in $\m{P}$.    

\begin{figure}[htb] 
\begin{minipage}[b]{1.0\linewidth} % A minipage that covers half the page, width-wise  
 \centerline{
\scalebox{0.50}{
 \includegraphics{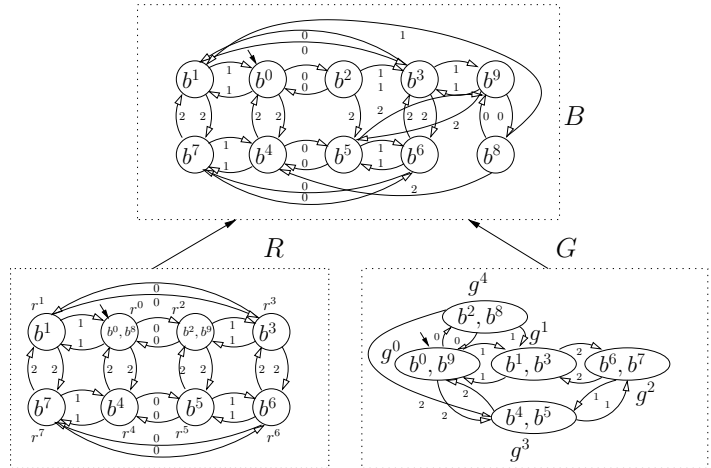}
} 
}
\end{minipage}  
\caption[ ]{Machine outside the closed partition set of $R$ in Fig. \ref{figStEvRedLattice}.}
\label{figExMacOutLat}
\end{figure}

Consider the example shown in Fig. \ref{figExMacOutLat}. Given the set of primaries $\{A,B,C\}$ shown
in Fig. \ref{figMainExample}, we want to determine if $G$ can correct one crash fault
among $\{A,B,C\}$. Since $G$ is outside the closed partition set of  $R$, we first construct $B$, which is the $\RCP$ of $G$ and $R$. 
Since $B$ is greater than both $R$ and $G$, we can determine how its
states are mapped to the states of $R$ and $G$ (similar to Fig. \ref{figStEvRedLattice}). For
example, $b^0$ and $b^8$ are mapped to $r^0$ in $R$, while $b^0$ and $b^9$ are mapped to
$g^0$ in $G$. Using this information, we can determine the mapping between the states of $R$ and $G$.
For example, since $b^0$ and $b^9$ are mapped to $r^0$ and  $r^2$ respectively, $g^0=\{r^0, r^2\}$.
Extending this idea, we get:
$$ g^1=\{r^1, r^3\}; g^2=\{r^6, r^7\};g^3=\{r^4, r^5\}; g^4=\{r^0, r^2\} $$   

In Fig. \ref{figFaultGraph} $(ii)$, the weakest edges of $G(\{A,B,C\})$ are $(r^0,r^1)$ and
$(r^2,r^3)$ (the other weakest edges not shown). Since $G$ separates all these edges, it can correct one crash fault among the
machines in $\{A,B,C\}$. However, note that, the machines in $\{A,B,C\}$ cannot correct a fault in $G$. For
example, if $G$ crashes and $R$ is in state $r^0$, we cannot determine if $G$ was in
state $g^0$ or $g^4$. This is clearly different from the case of the fusion machines presented in
this paper, where faults could be corrected among \emph{both} primaries and backups.
%(like $M_2$ or $F_2$ in Fig. \ref{figStEvRedLattice}).   
 
%\begin{observation}
%Consider a set of machines $\m{P}$ and a set of fault tolerant backups $\m{G}$, where each machine in $\m{G}$ is
%outside the closed partition set of $\RCP(\m{P})$. While $\m{G}$ can correct faults among
%$\m{P}$, $\m{P}$ cannot correct faults in $\m{G}$. 
%\end{observation}
%

\section{Related Work}\label{secRelatedWork}

%Our work in \cite{BharOgaleGargFusion07} introduces the concept of the fusion of finite state
%machines. That work deals with the special case where the number of backup machines is equal to the
%number of faults and presents a brute force exponential-time algorithm to generate the equivalent of
%the \fusion{1}{1} idea in this paper. The system model in that paper only allowed crash faults. In
%this paper, we generalize that idea to \fusion{f}{m}, where we examine the problem of tolerating $f$
%faults using  $m$ additional machines. Our system model allows both crash and Byzantine faults and
%we present polynomial time algorithms for generating the backups.
%%We introduce the concept of the fault graphs associated with such machines and define a theory to
%formally understand the concept of fusion. Using this theory, we present an efficient algorithm for
%generating the smallest set of backup machines, to correct $f$ faults in a given set of machines.
Our work in \cite{BharOgaleGargFusion07} introduces the concept of the fusion of \FSMs{}, and
presents an algorithm to generate a backup to correct one crash fault among a given set of
machines. This paper  is based on our work in \cite{OgaBhar09,bharGargFsmOpodis2011}. The work presented in 
\cite{GarOgaFusibleDS,balaGargFusData2011,GarBeyondReplication2010} explores fault
tolerance in distributed systems  with programs hosting large data structures. The key idea there is
to use erasure/error correcting codes \cite{BerleCoding68} to reduce the space overhead of
replication. Even in this paper, we exploit the similarity between fault tolerance in \FSMs{} and fault tolerance in a
block of bits using erasure codes in section \ref{secGraphHamming}. However, there is one important difference between erasure codes involving bits
and the \FSM{} problem.  In erasure codes, the value of the redundant bits depend on the data bits.
In the case of \FSMs{}, it is not feasible to transmit the state of all the machines after each
event transition to calculate the state of the backup machines. Further, recovery in such an
approach is costly due to the cost of decoding. In our solution, the backup machines 
act on the same inputs as the original machines and independently transition to
suitable states. Extensive work has been done \cite{HuffSynth,HopTechReport71} on the minimization of completely
specified \FSMs{}, but the minimized machines are equivalent to the original machines. In our
approach, we reduce the $\RCP$ to generate efficient backup machines that are lesser than the $\RCP$. Finally, since we assume a trusted recovery agent, the work on
consensus in the presence of Byzantine faults \cite{LaSh82,PeaseLamp80}, does not apply to our
paper.

%Further, The famous FLP result \cite{FLP85} states that it is impossible to achieve consensus among a given
%set of machines in an asyncrhonous system with even one faulty  machine. Our system model does not
%assume asynchrony and the state of all the faulty machines are available for recovery.  Similarly,
%there has been a considerable body of research for solving consensus among \FSMs{}, in  synchronous
%systems. Assume a system of $n$ machines in which upto $f$ machines may be faulty. As far as
%Byzantine faults are concerned, we cannot achieve consensus  unless $n > 3f$ \cite{PeaseLamp80}. In
%the case of crash faults, it has been shown that it will require at least $f+1$ synchronous rounds
%to achieve consensus \cite{LampFisch82}. These results do not apply to this paper because, even
%though the individual machines may fail, it is assumed that the recovery system, which can detect
%and correct faults, is not faulty.  
%

\section{Conclusion}\label{secConc}
We present a fusion-based solution to correct $f$ crash or $\lfloor f/2 \rfloor$ Byzantine faults
among $n$ \FSMs{} using just $f$ backups as compared to the traditional approach of replication
that requires $nf$ backups.
In table \ref{tableComparison}, we summarize our results and compare the various
parameters for replication and
fusion. In this paper, we present a framework to understand fault
tolerance in machines and provide an algorithm that generates backups that are optimized for states
as well as events. Further, we present algorithms for detection and the correction of faults with
minimal overhead over replication. 

Our evaluation of fusion over standard benchmarks shows that efficient backups exist for many
examples. To illustrate the practical use of
fusion, we describe a fusion-based design of a distributed application in the MapReduce framework. While the current replication-based solution may require 1.8 million map tasks, a fusion-based solution
requires just 1.4 million map tasks with minimal overhead in terms of time as compared to replication.
This can result in considerable savings in space and other computational resources such as power.

In the future, we wish to implement the design presented in section \ref{secMapReduce} using the
Hadoop framework \cite{hadoop} and 
compare the end-to-end performance of replication and our fusion-based solution. In particular we wish to focus on the
space incurred by both solutions, the time and computation power taken for a set of tasks to complete with and without
faults. Further, we wish to explore the
existence of efficient backups if we allow information exchange among the primaries. Finally, we
wish to design efficient algorithms to generate backups both inside and outside the closed partition
set of the $\RCP$.

\bibliographystyle{plain}
\bibliography{refs_ipdps09,refs_disc10} 
\clearpage
\appendix 

\setcounter{theorem}{0}

\begin{figure*}[htb]
\centerline{
\scalebox{0.5}{
 \includegraphics{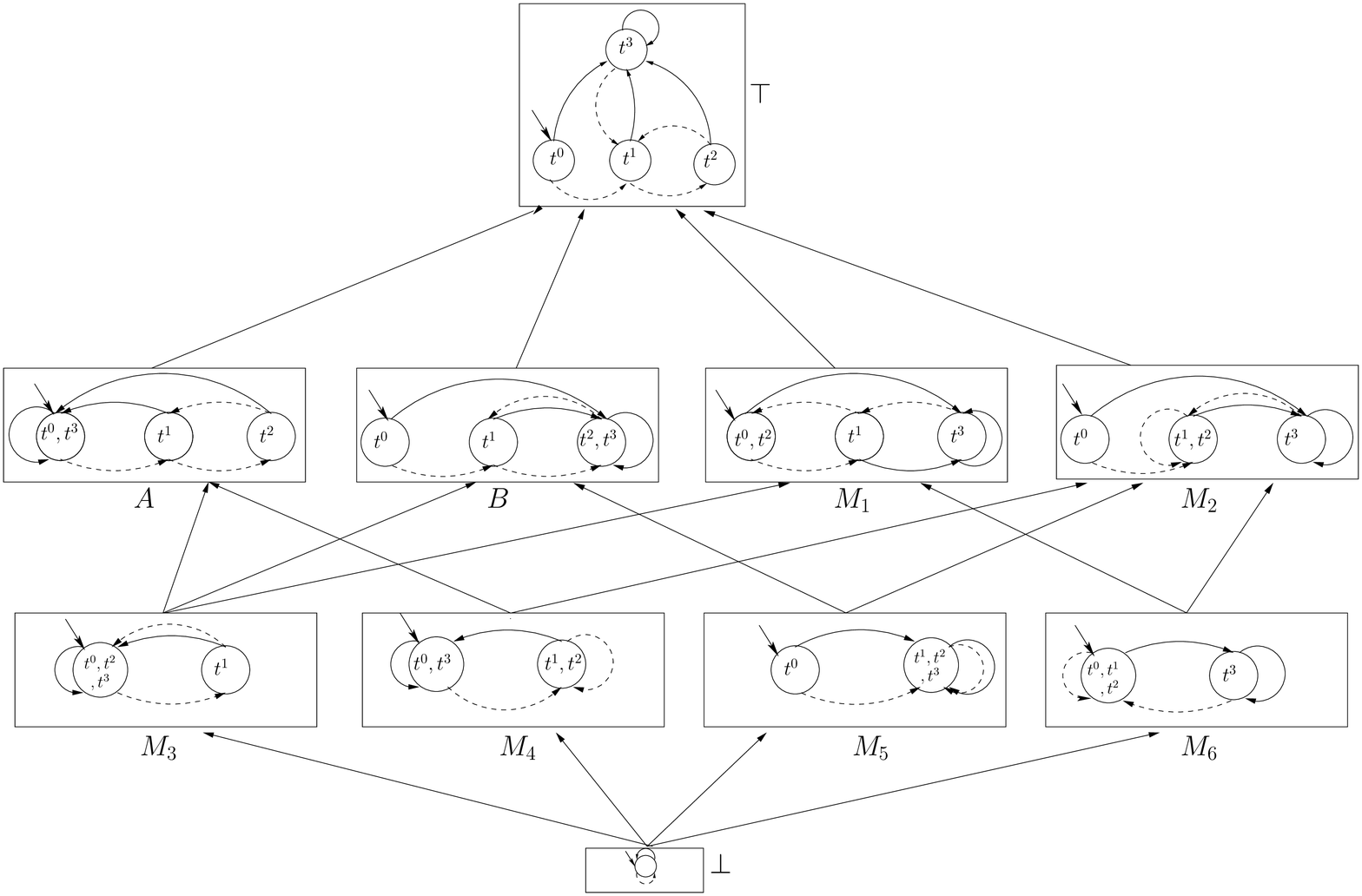}
} 
}
\caption{Closed partition set for the $\RCP$ of $\{A,B\}$.}
\label{figOldClosedPartitionLattice}
\end{figure*}

\begin{figure*}[!htb]
\centerline{
\scalebox{0.4}{
 \includegraphics{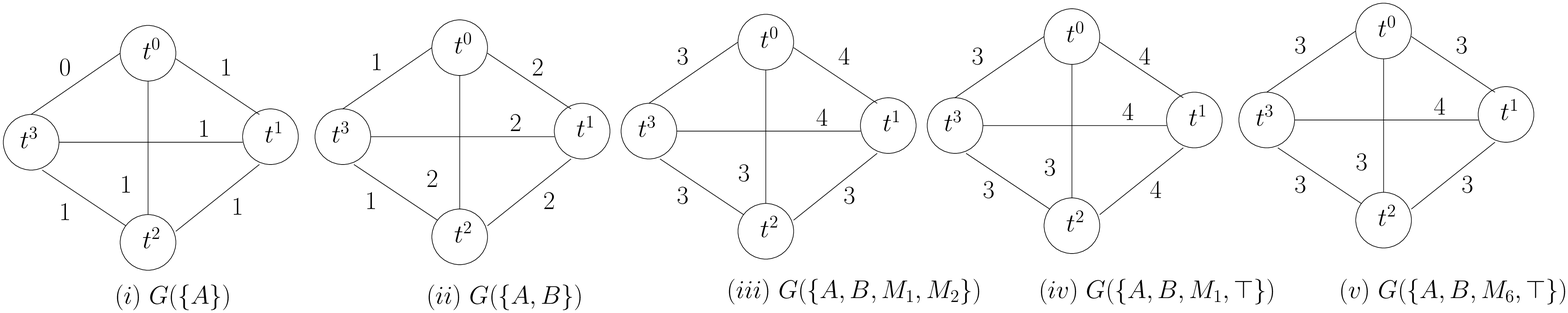}
} 
}
\caption{Fault Graphs for sets of machines shown in Fig. \ref{figOldClosedPartitionLattice}.}
\label{figOldFaultGraph}
\end{figure*}

\section{Event-Based Decomposition of Machines}\label{secAppEvReduction}

\begin{figure}[ht]
\begin{minipage}[b]{0.8\linewidth}
\centering
\includegraphics[scale=0.40]{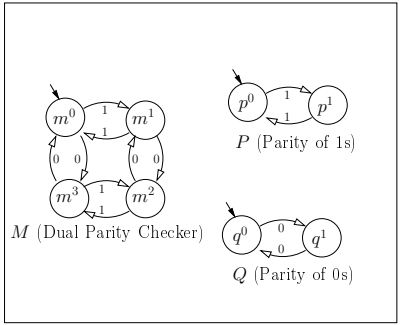}
\caption{Event-based decomposition of a machine.}
\label{figExampleEvReduce}
\end{minipage}
\end{figure}

In this section, we ask a question that is fundamental to the understanding of \FSMs{}, independent of
fault-tolerance: Given a machine $M$, can it be \emph{replaced} by two or more machines executing in
parallel, each containing fewer
events than $M$? In other words, given the state of these fewer-event machines, can we uniquely
determine the state of $M$? In Fig. \ref{figExampleEvReduce}, the 2-event machine $M$ (it  contains events 0 and
1 in its event set), checks for the parity of 0s \emph{and} 1s. $M$ can be replaced by two 1-event machines $P$ and $Q$,
that check for the parity of just 1s or 0s respectively. Given the state
of $P$ and $Q$, we can determine the state of $M$. 
In this section, we explore the problem of replacing a given machine $M$ with two or more machines,
each containing fewer events than $M$. We present an algorithm to generate such event-reduced
machines with time complexity polynomial in the size of $M$. This is 
important for applications with limits on the number of events each
individual process running a \FSM{} can service. We first define the notion of event-based decomposition. 

\begin{definition}
A \emph{\decomp{k}{\e}} of a machine $M$ $(X_M$, $\alpha_M$, $\Sigma_M$, $m^0)$ is a set of $k$
machines $\mathcal{E}$, each less than $M$,
such that $\w(M,\mathcal{E})>0$ and $\forall P (X_P,\alpha_P,\Sigma_P,p^0)\in \mathcal{E}$,
$|\Sigma_P| \leq |\Sigma_M|-\e$. 
\end{definition}

As $\w(M,\mathcal{E})>0$, given the state of the machines in $\mathcal{E}$,
the state of $M$ can be determined. So,
the machines in $\mathcal{E}$, each containing at most $|\Sigma_M|-e$ events, can effectively replace $M$. 
In Fig. \ref{figEventReduce}, we present the \emph{eventDecompose} algorithm that takes as input,
machine $M$, parameter $\e$, and returns a \decomp{$k$}{$\e$} of $M$ (if it exists) for some $k \leq
|X_M|^2$. 

In each iteration, Loop 1 generates machines that contain at least one event less than the
machines of the previous iteration. So, starting with $M$ in the first iteration, at the end of
$e$ iterations, $\mathcal{M}$ contains the set of largest machines less than $M$, each containing at most $|\Sigma_M|-\e$ events. 

Loop 2, iterates through each machine $P$ generated in the previous iteration, and uses the
\emph{reduceEvent} algorithm (same as the algorithm presented in Fig. \ref{figFusionAlgo}) to generate the set of largest machines less than $P$ containing at least one
event less than $\Sigma_P$. To generate a machine less than $P$,
that does not contain an event $\sigma$ in its event set, the \emph{reduceEvent} algorithm
combines the states such that they loop
onto themselves on $\sigma$. The algorithm then constructs the largest machine that contains these states in the combined form. This
machine, in effect, ignores $\sigma$. This procedure is repeated for all events in $\Sigma_P$ and the
largest incomparable machines among them are returned. 
Loop 3 constructs an event-decomposition $\mathcal{E}$ of $M$, by iteratively adding at least one machine from
$\mathcal{M}$ to separate each pair of states in $M$, thereby ensuring that $\w(\mathcal{E})> 0$. 
Since each machine added to $\mathcal{E}$ can separate more than one pair of states, an
efficient way to implement Loop 3 is to check for the pairs that still need to be separated in
each iteration and add machines till no pair remains. 

\begin{figure*}[htb]  
\begin{minipage}[b]{0.5\linewidth} % A minipage that covers half the page, width-wise  
 \centerline{
\scalebox{0.38}{
 \includegraphics{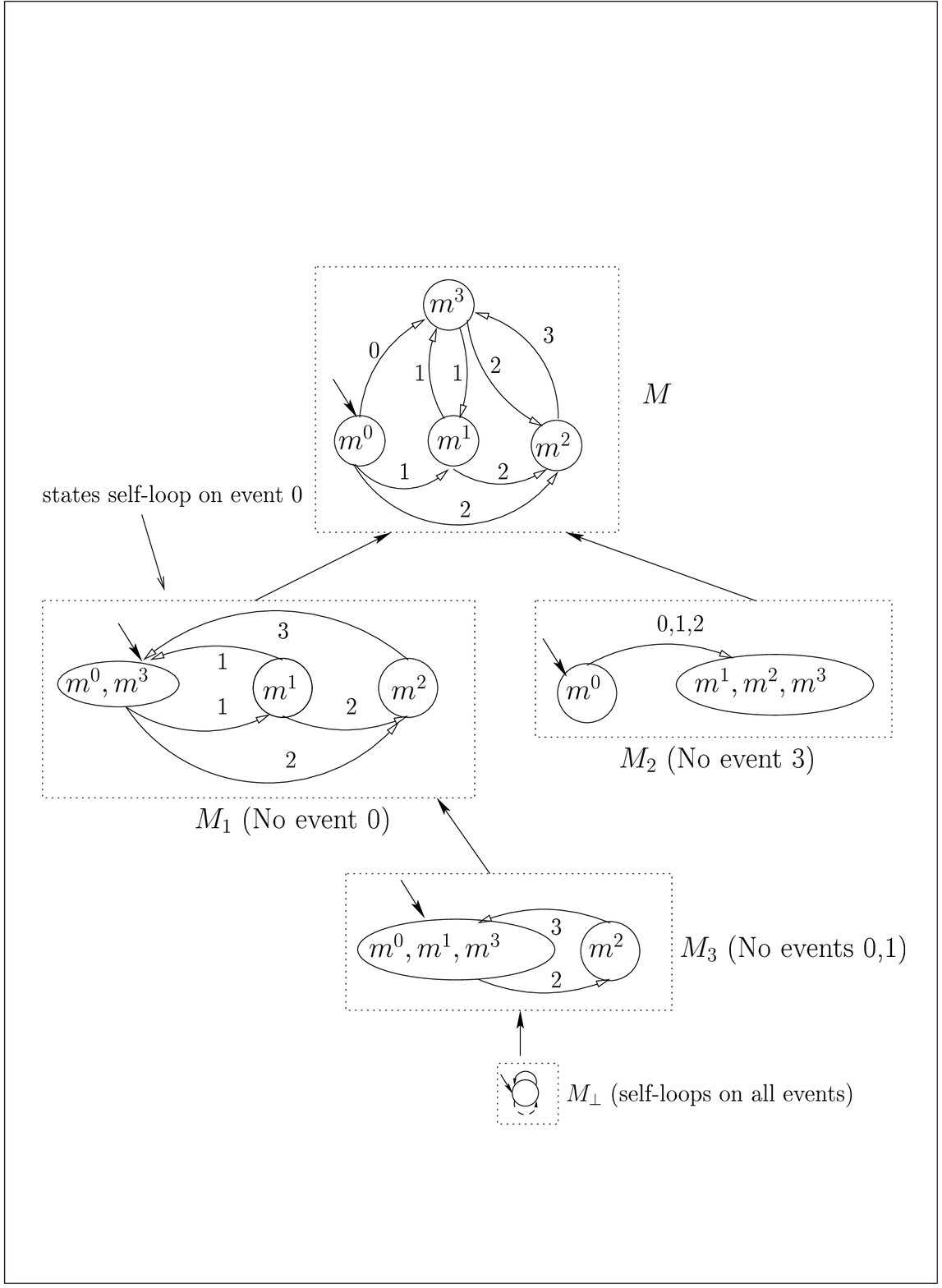}
} 
}
\end{minipage}  
\hspace{0.1cm} % To get a little bit of space between the figures  
\fbox{\begin{minipage}[b]  {0.4\linewidth}
{\small
\emph{eventDecompose}\\
    \h {\bf Input}: Machine $M$ with state set $X_M$, event set $\Sigma_M$\\
		\h and transition function $\alpha_M$;\\ 
    \h {\bf Output}: \decomp{$k$}{$\e$} of $M$ for \\
    \h some $k \leq |X_M|^2$;\\
     \h $\mathcal{M} = \{M\}$; \\
     \h{\bf for} {($j=1$ to $\e$)} //Loop 1\\
	\h\h 	$\mathcal{G} \leftarrow \{\}$;\\
	\h\h {\bf for} {$(P \in \mathcal{M})$} //Loop 2\\
		\h\h\h $\mathcal{G}= \mathcal{G} \cup \emph{reduceEvent}(P)$;\\
	\h\h $\mathcal{M}$ = $\mathcal{G}$; \\
    \h $\mathcal{E} \leftarrow \{\}$;\\
    \h {\bf for} $(m_i,m_j \in X_M)$ //Loop 3\\
	\h\h {\bf if} ($\exists E \in \m{M}: E$ separates $m_i,m_j$)\\
	\h\h\h $\mathcal{E} \leftarrow \mathcal{E} \cup \{E\}$;\\
	\h\h{\bf else}\\     
	\h\h\h{\bf return} $\{\}$;\\
\h {\bf return} $\mathcal{E}$;\\

\emph{reduceEvent}\\
    \h {\bf Input}: Machine $P$ with state set $X_P$, event set $\Sigma_P$\\
		\h and transition function $\alpha_P$;\\ 
    \h {\bf Output}: Largest Machines $< P$ with $\leq |\Sigma_P|-1$ events;\\
    \h $\mathcal{B} = \{\}$; \\
    \h {\bf for} $(\sigma \in \Sigma_P)$\\ 
	\h\h Set of states, $X_B = X_P$;\\ 
	\h\h //combine states to self-loop on $\sigma$\\
	\h\h {\bf for} ($s \in X_B$)\\		
		\h\h\h $s = s \cup \alpha_P(s,\sigma)$;\\
	\h\h $\mathcal{B} = \mathcal{B} \cup \{$Largest machine consistent with $X_B\}$; \\
    \h {\bf return} largest incomparable machines in $\mathcal{B}$;

}
\end{minipage}
} % end \fbox
\caption[ ]{Algorithm for the event-based decomposition of a machine.}
\label{figEventReduce}
\end{figure*}

Let  the 4-event machine $M$ shown in Fig. \ref{figEventReduce} be the input to the \emph{eventDecompose} algorithm with $\e=1$. 
In the first and only iteration of Loop 1, $P=M$ and the \emph{reduceEvent} algorithm
generates the set of largest 3-event machines less than $M$, by successively eliminating each event. 
To eliminate event 0, since $m^0$ transitions to $m^3$ on event $0$, these two states are combined.
This is repeated for all states and the largest machine containing all the combined states self
looping on event 0 is $M_1$. Similarly, the largest machines not acting on events 3,1 and 2 are $M_2$,
$M_3$ and $M_\bot$ respectively. The \emph{reduceEvent} algorithm returns $M_1$ and $M_2$ as the
only largest incomparable machines in this set. The
\emph{eventDecompose} algorithm returns $\mathcal{E}=\{M_1$, $M_2\}$, since each
pair of states in $M$ are separated by $M_1$ or $M_2$. Hence, the 4-event $M$ can be replaced by the 3-event 
$M_1$ and $M_2$, i.e., $\mathcal{E}=\{M_1,M_2\}$ is a \decomp{2}{1} of $M$. 

\begin{theorem}
Given machine $M$ $(X_M,\alpha_M,\Sigma_M,m^0)$, the \emph{eventDecompose} algorithm generates a
\decomp{$k$}{$\e$} of $M$ (if it exists) for some  $k \leq |X_M|^2$. 
\end{theorem}

\begin{proof} The \emph{reduceEvent} algorithm 
exhaustively generates the largest incomparable machines that ignore at least one event in $\Sigma_M$. After
$e$ such reduction in events, Loop 3 selects one machine (if it exists) among $\mathcal{M}$ to
separate each pair of states in $X_M$. This ensures that at the end of Loop 3, either
$\w(\mathcal{E})>0$ or the algorithm has returned $\{\}$ (no \decomp{$k$}{$e$} exists). Since
there are at most $|X_M|^2$ pairs of states in $X_M$, there are at most $|X_M|^2$ iterations of
Loop 3, in which we pick one machine per iteration. Hence, $k \leq |X_M|^2$. 
\end{proof}

The \emph{reduceEvent} algorithm visits each state of machine $M$ to create blocks of states
which loop to the same block on event $\sigma \in \Sigma_M$. This has time complexity $O(|X_M|)$ per
event. The cost of
generating the largest closed partition corresponding to this block is $O(|X_M| |\Sigma_M|)$ per
event. Since
we need to do this for all events in $\Sigma_M$, the time complexity to reduce at least one event is
$O(|X_M| |\Sigma_M|^2)$. In the \emph{eventDecompose} algorithm, the first iteration 
generates at most $|\Sigma_M|$ machines, the second iteration at most $|\Sigma_M|^2$ machines and the $\e^{th}$
iteration will contain $O(|\Sigma_M|^\e)$ machines. The rest of the analysis is similar to the one
presented in section \ref{secFsmTc} and the time complexity of the \emph{reduceEvent} algorithm is $O(|X_M|
|\Sigma_M|^{\e+1})$.

%The time complexity to reduce at most one event
%from any machine is $O(|X_M| |\Sigma_M|^2)$. Hence, the time complexity to generate the set of
%event-reduced machines in $\mathcal{M}$ at the end of $\e$ iterations, i.e, the time complexity of
%Loop 1 is $O(|X_M|\cdot
%|\Sigma_M|^2(1+|\Sigma_M|+|\Sigma_M|^2+\ldots+|\Sigma_M|^{\e-1}))$ which reduces to $O(|X_M|
%|\Sigma_M|^2(\frac{|\Sigma_M|^{\e}-1}{|\Sigma_M|-1}))$ = $O(|X_M|
%|\Sigma_M|^{\e+1})$.  
%
To generate the \decomp{$k$}{$\e$} from the set of machines in
$\mathcal{M}$, we find a machine in $\mathcal{M}$ to separate each pair of states in $X_M$.
Since there are $O(|X_M|^2)$ such pairs, the number of iterations of Loop 3 is $O(|X_M|^2)$. In each
iteration of Loop 3, we find a machine among the $O(|\Sigma_M|^{\e})$ machines of $\mathcal{M}$ that
separates a pair $m_i,m_j \in X_M$. To check if a machine separates a pair of states just takes
$O(|X_M|)$ time. Hence the time complexity of Loop 3 is $O(|X_M|^3|\Sigma_M|^{\e})$. So, the overall
time complexity of the \emph{eventDecompose} algorithm is the sum of the time complexities of Loop 1
and 3, which is $O(|X_M||\Sigma_M|^{\e+1}+|X_M|^3|\Sigma|^{\e})$.  

\section{Incremental Approach to Generate Fusions} \label{secAppIncFusion}

\begin{figure} \begin{center}
\fbox{
\begin{minipage}[b]  {0.90\linewidth}
{\small
\emph{incFusion}\\
\h{\bf Input}: Primaries $\mathcal{P} =\{P_1, P_2, \ldots P_n\}$, faults $f$, \\ 
\h state-reduction parameter $\State$, event-reduction parameter $\Event$;\\
\h{\bf Output}: \fusion{f}{f} of $\mathcal{P}$; \\
\h$\mathcal{F} \leftarrow \{P_1\}$; \\
%\h$\mathcal{N} \leftarrow \{\}$; \\
    \h {\bf for} ($i=2$ to $n$) \\ 
%	\h\h //input for next iteration of fusion algorithm;\\
	\h\h $\mathcal{N} \leftarrow \{P_i\} \cup \RCP(\mathcal{F})$;\\
	\h\h $\mathcal{F} \leftarrow \emph{genFusion}(\mathcal{N},f, \State,\Event)$;\\
    \h {\bf return} $\mathcal{F}$;
}\end{minipage}
}
\end{center}
\caption{Incremental fusion algorithm.}
\label{figIncFusionAlgo}
\end{figure}

In Fig. \ref{figIncFusionAlgo}, we present an incremental approach to generate the fusions, referred
to as the \emph{incFusion} algorithm, in
which we may never have to reduce the $\RCP$ of all the primaries. In each iteration, we generate the
fusion corresponding to a new primary and the $\RCP$ of the (possibly small) fusions  generated for the set of primaries
in the previous iteration. 

\begin{figure*}[htb] 
\centerline{ 
\scalebox{0.50}{ 
\includegraphics{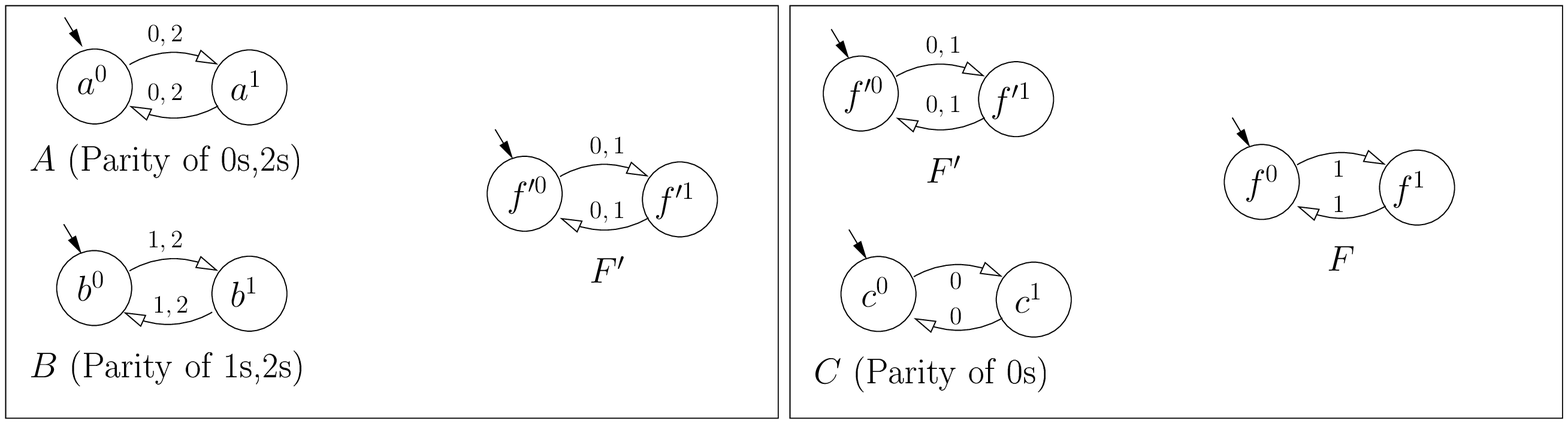} } } 
\caption{Incremental Approach: first generate $F'$ and then $F$.}
\label{figIncremental} \end{figure*}

In Fig. \ref{figIncremental}, rather than generate a fusion by
reducing the 8-state $\RCP$ of $\{A,B,C\}$, we can reduce the 4-state $\RCP$ of $\{A,B\}$ to
generate fusion $F'$ and then reduce the 4-state $\RCP$ of
$\{C,F'\}$ to generate fusion $F$. In the following paragraph, we present the proof of correctness for
the incremental approach and show that it has
time complexity $O(\rho^n)$ times better than that of the \emph{genFusion} algorithm, where $\rho$
is the average state reduction achieved by fusion. 

\begin{theorem}
Given a set of $n$ machines $\mathcal{P}$, the
\emph{incFusion} algorithm generates an \fusion{f}{f} of $\mathcal{P}$.
\end{theorem}
\begin{proof} 
We prove the theorem using induction on the variable $i$ in the algorithm. For the base case, i.e.,
$i=2$, $\m{N}=\{P_1,P_2\}$ (since $RCP(\{P_1\})=P_1$). Let the \fusion{f}{f} generated by the
\emph{genFusion} algorithm for 
$\m{N}=\{P_1,P_2\}$ be denoted
$\mathcal{F}^1$. For $i=3$, let the \fusion{f}{f} generated for $\m{N}=\{P_3, \RCP(\mathcal{F}^1)\}$ be
denoted $\mathcal{F}^2$. We show that $\mathcal{F}^2$ is an \fusion{f}{f} of $\{P_1,P_2,P_3\}$.
Assume $f$ crash faults among $\{P_1P_2,P_3\} \cup \m{F}^2$. Clearly, less than or equal to $f$ machines in $\{P_3\} \cup
\mathcal{F}^2$ have crashed. Since $\m{F}^2$ is an \fusion{f}{f} of $\{P_3, 
\RCP(\m{F}^1)\}$, we can generate the state of all the machines in $\RCP(\m{F}^1)$ 
and the state of the crashed machines among $\{P_3\} \cup
\mathcal{F}^2$. Similarly, less than or equal to $f$ machines have crashed among $\{P_1,P_2\}$.
Hence, using the state of the available machines among $\{P_1,P_2\}$ and
the states of all the
machines in $\mathcal{F}^1$ we can generate the state of the crashed machines among $\{P_1,P_2\}$. 

Induction Hypothesis: Assume that the set of machines $\mathcal{F}^i$, generated in iteration $i$, is an
\fusion{f}{f} of $\{P_1 \ldots P_{i+1}\}$.  Let the \fusion{f}{f} of
$\{P_{i+2}, \RCP(\mathcal{F}^i)\}$ generated in iteration $i+1$ be denoted $\m{F}^{i+1}$. To prove: $\mathcal{F}^{i+1}$ is an \fusion{f}{f} of $\{P_1
\ldots P_{i+2}\}$. The proof is similar to that for the base case. Using the state of the available machines
in $\{P_{i+2}\} \cup
\mathcal{F}^{i+1}$, we can generate the state of all the machines in $\mathcal{F}^{i}$ and $\{P_{i+2}\} \cup
\mathcal{F}^{i+1}$. Subsequently, we can generate 
the state of the crashed machines in $\{P_1 \ldots P_{i+1}\}$. 
\end{proof}

From observation \ref{obGenFusTc}, the \emph{genfusion} algorithm has time complexity,\\
$O(fN^4|\Sigma|+fN^5)$ (assuming $\State=0$ and $\Event=0$ for simplicity). Hence, if the size of
$\m{N}$ in the $i^{th}$
iteration of the \emph{incFusion} algorithm is denoted by $N_i$, then the time complexity of the
\emph{incFusion} algorithm, $T_{inc}$ is given by the
expression $\Sigma_{i=2}^{i=n}O(fN_i^4|\Sigma|+fN_i^5)$. 

Let the number of states in each primary be $s$. For $i=2$, the primaries are $\{P_1,P_2\}$ and
$N_1=O(s^2)$. For $i=3$, the primaries are \{$\RCP(\m{F}^1),P_3\}$. Note that $\RCP({\m{F}^1)}$ is also a fusion machine. Since
we assume an average reduction of $\rho$ (size of $\RCP$ of primaries/average size of each fusion),
the number of states in $\RCP(\m{F}^1)$ is $O(s^2/\rho)$. So , $N_2=O(s^3/\rho)$. Similarly,
$N_3=O(s^4/\rho^2)$ and $N_{i}=O(s^{i+1}/\rho^{i-1})$. So,
$$T_{inc}=O(|\Sigma|f\Sigma_{i=2}^{i=n}s^{4i+4}/\rho^{4i-4}+f\Sigma_{i=2}^{i=n}s^{5i+5}/\rho^{5i-5})$$
$$= O(|\Sigma|fs^4\rho^4\Sigma_{i=2}^{i=n}(s/\rho)^{4i}+fs^5\rho^5\Sigma_{i=2}^{i=n}(s/\rho)^{5i})$$ 
This is the sum of a geometric progression and hence, $$T_{inc} =
O(|\Sigma|fs^4\rho^4(s/\rho)^{4n}+fs^5\rho^5(s/\rho)^{5n})$$ 
Assuming $\rho$ and $s$ are constants,
$T_{inc}=O(f|\Sigma|s^n/\rho^n+fs^n/\rho^n)$. Note that, the time complexity of the
\emph{genFusion} algorithm in Fig. \ref{figFusionAlgo} is $O(f|\Sigma|s^n+ fs^n)$. Hence, 
the \emph{incFusion} algorithm achieves $O(\rho^n)$  savings
in time complexity over the $\emph{genFusion}$ algorithm.

\end{document}